\documentclass[manuscript,screen,acmtog]{acmart}

\usepackage{empheq}
\AtBeginDocument{%
  }

\setcopyright{acmlicensed}
\copyrightyear{2018}
\acmYear{2018}
\acmDOI{XXXXXXX.XXXXXXX}

\acmConference[Conference acronym 'XX]{Make sure to enter the correct
  conference title from your rights confirmation email}{June 03--05,
  2018}{Woodstock, NY}
\acmISBN{978-1-4503-XXXX-X/18/06}

\acmSubmissionID{560}

\usepackage{algorithm}
\usepackage{algpseudocode}
\usepackage{subfigure} 
\usepackage{mdframed}
\usepackage{xfrac}

\definecolor{myblue}{RGB}{230, 240, 255} % Light blue background
\definecolor{myborder}{RGB}{50, 100, 200} % Darker blue border

\mdfdefinestyle{mdfexample1}{
    backgroundcolor=gray!10,
    linecolor=white,
    linewidth=2pt,
    frametitlebackgroundcolor=gray!20,
    frametitlerule=true,
    frametitlerulewidth=1pt,
    frametitleaboveskip=5pt,
    frametitlebelowskip=5pt,
    nobreak=true,
}
\begin{document}

\title{Monte Carlo PDE Solvers for Nonlinear Radiative Boundary Conditions}

\renewcommand{\shortauthors}{Bao et al.}
\newcommand{\changed}[1]{{#1}}
\newcommand{\changedd}[1]{{#1}}

\author{Anchang Bao}
\affiliation{%
  \institution{School of Software, BNRist, Tsinghua University}
  \city{Beijing}
  \country{China}
}
\email{baoanchang02@gmail.com}

\author{Enya Shen}
\affiliation{%
  \institution{School of Software, BNRist, Tsinghua University}
  \city{Beijing}
  \country{China}
  }
\affiliation{%
  \institution{Haihe Lab of ITAI}
  \city{Tianjin}
  \country{China}
  }

\email{shenenya@tsinghua.edu.cn}

\author{Jianmin Wang}
\affiliation{%
  \institution{School of Software, BNRist, Tsinghua University}
  \city{Beijing}
  \country{China}
  }
\email{jimwang@tsinghua.edu.cn}


\begin{abstract}
    Monte Carlo PDE solvers have become increasingly popular for solving heat-related partial differential equations in geometry processing and computer graphics due to their robustness in handling complex geometries. While existing methods can handle Dirichlet, Neumann, and linear Robin boundary conditions, nonlinear boundary conditions arising from thermal radiation remain largely unexplored.

    In this paper, we introduce a Picard-style fixed-point iteration framework that enables Monte Carlo PDE solvers to handle nonlinear radiative boundary conditions. \changedd{While strict theoretical convergence is not generally guaranteed, our method remains stable and empirically convergent with a properly chosen relaxation coefficient. Even with imprecise initial boundary estimates, it progressively approaches the correct solution.} Compared to standard linearization strategies, the proposed approach achieves significantly higher accuracy.
    
    To further address the high variance inherent in Monte Carlo estimators, we propose a heteroscedastic regression-based denoising technique specifically designed for on-boundary solution estimates, filling a gap left by prior variance reduction methods that focus solely on interior points. We validate our approach through extensive evaluations on synthetic benchmarks and demonstrate its effectiveness on practical heat radiation simulations with complex geometries.
\end{abstract}





\received{20 February 2007}
\received[revised]{12 March 2009}
\received[accepted]{5 June 2009}

\begin{teaserfigure}
    \centering
    \includegraphics[width=\linewidth]{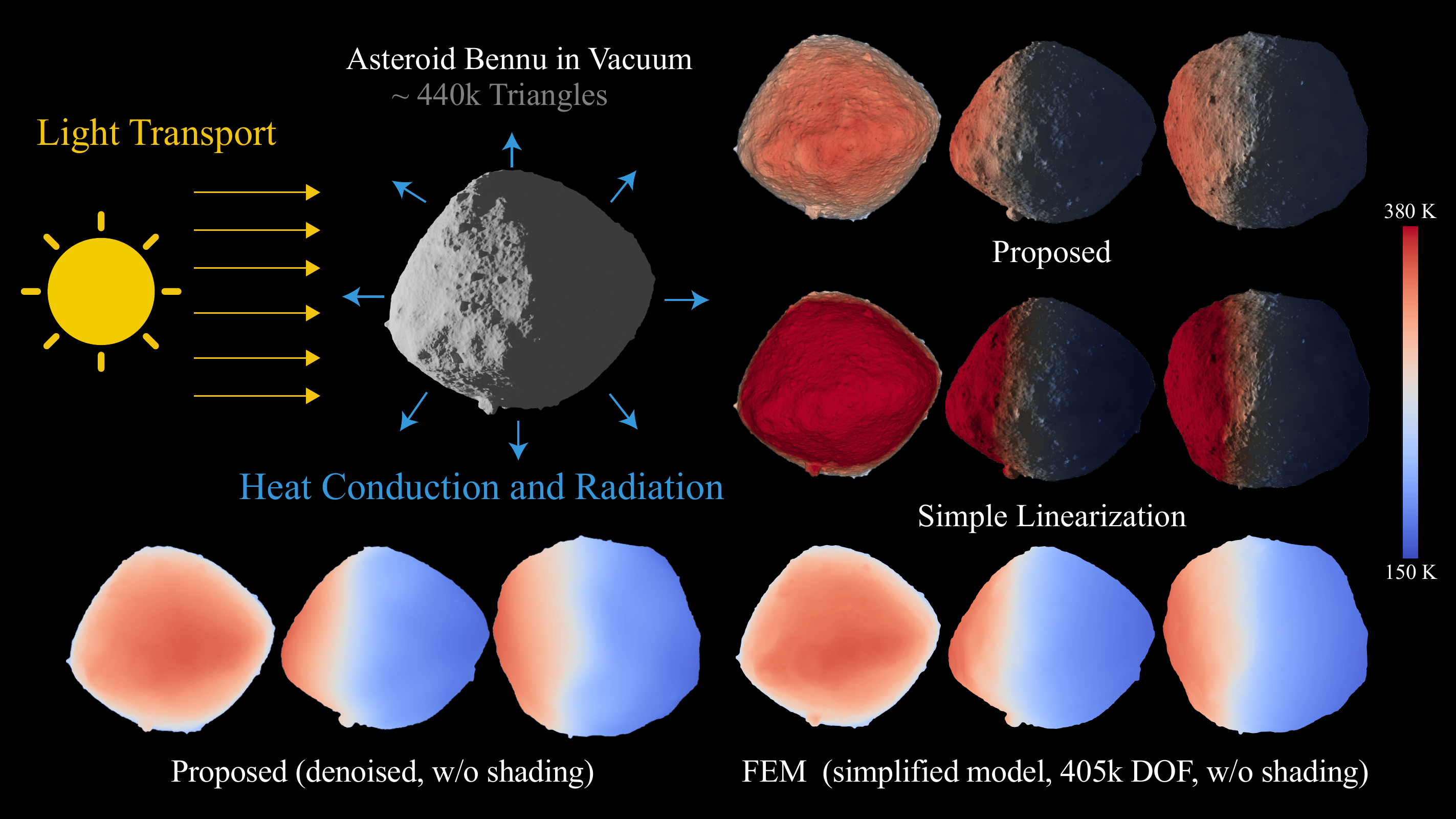}
    \caption{We propose a fixed-point Monte Carlo framework for solving PDEs with nonlinear boundary conditions. Shown here is an application to \changed{a} steady-state light--heat coupled \changed{problem} on an asteroid in vacuum. A naive linearization of the radiative boundary condition, combined with an imprecise initial temperature, produces \changed{systematically biased} surface temperatures. In contrast, the proposed Picard-style iteration robustly resolves the nonlinear coupling between radiative transfer and heat conduction, yielding physically consistent temperature distributions under complex self-shadowing and geometry. }
    \label{fig:coupled}
\end{teaserfigure}

\maketitle

\section{Introduction}

% Thermal analysis is a fundamental component of many engineering and scientific applications, including aerospace design, planetary science, and computer graphics. 
Recent advances in Monte Carlo PDE solvers \cite{miller2024walkin} have demonstrated that stochastic methods provide a robust and geometry-flexible alternative to classical mesh-based solvers, particularly for problems involving highly complex or irregular boundaries. Prior work has successfully applied these solvers to a variety of heat-related problems, including infrared rendering~\cite{bati2023coupling} and interactive shape crafting~\cite{de2023heat}.

Heat transfer is governed by three physical mechanisms: conduction, convection, and radiation. Heat conduction in homogeneous materials is described by the Laplace equation, while convective heat exchange with the environment is typically modeled through Robin boundary conditions. Existing Monte Carlo PDE solvers, such as Walk-on-Spheres and Walk-on-Boundary methods, can handle a wide range of linear boundary conditions, including Dirichlet \cite{sawhney2020monte}, Neumann \cite{sawhney2023walk}, and Robin conditions \cite{miller2024walkin, sugimoto2023practical}. As a result, conduction and convection phenomena are well supported within the current Monte Carlo PDE solver framework.

However, radiative heat transfer remains largely unaddressed in Monte Carlo PDE solvers. Unlike convection, radiation introduces a fundamentally nonlinear boundary condition of the form
\begin{equation}
    \frac{\partial T}{\partial \vec{n}} + \frac{\sigma \epsilon}{k_0} \left( T^4 - T_0^4 \right) = 0,
\end{equation}
where $T$ denotes the surface temperature, $T_0$ the ambient temperature, $\epsilon$ the emissivity, $k_0$ the thermal conductivity, and $\sigma$ the Stefan-Boltzmann constant. This nonlinear dependence on the unknown temperature breaks the linear structure that underlies classical Monte Carlo estimators, making direct recursive integration intractable.

In this work, we address this limitation by integrating a fixed-point iteration scheme into Monte Carlo PDE solvers, enabling the efficient and robust simulation of steady-state radiative heat transfer on complex geometries. Our approach preserves the geometric flexibility and mesh-free nature of Monte Carlo solvers while extending their applicability to nonlinear radiative boundary conditions. In addition, we observe that radiative problems often suffer from severe noise at boundary evaluation points, a case not covered by existing variance reduction techniques. To address this issue, we introduce a heteroscedastic regression-based denoising strategy that operates directly on boundary estimates and can be seamlessly embedded into the fixed-point iteration process.

% Our main contributions are as follows:
% \begin{itemize}
%     \item We propose a fixed-point iteration framework for solving nonlinear radiative boundary value problems within Monte Carlo PDE solvers, extending Monte Carlo PDE solution estimators to full heat transfer models.
%     \item We introduce a heteroscedastic regression-based denoising method that effectively handles noisy boundary evaluations and integrates naturally with the iterative solver.
%     \item We validate our approach on complex geometries and physically meaningful heat radiation scenarios, demonstrating competitive accuracy and robustness compared to \changed{the} vanilla linearization \changed{strategy} and finite element methods.
% \end{itemize}

\section{Related Work}

\subsection{Monte Carlo PDE Solvers}

Monte Carlo methods for solving partial differential equations originate from potential theory and probability theory, where solutions to elliptic boundary value problems admit probabilistic representations. A foundational result due to Kakutani~\cite{kakutani1944143} establishes that the solution to the Laplace equation with Dirichlet boundary conditions can be expressed as the expectation of the boundary values evaluated at the first-hitting point of a Brownian motion:
\begin{equation}
    u(x) = \mathbb{E}_x\left[ u(X_\tau) \right],
\end{equation}
\changed{where $X_t$ denotes Brownian motion starting at $x$, $\tau$ is the first time the process reaches the boundary $\partial \Omega$, \changedd{and} $X_{\tau}$ is the exit location\changedd{.}}

Direct simulation of Brownian motion via time discretization is computationally inefficient. \citet{muller1956some} introduced the Walk-on-Spheres (WoS) algorithm, which exploits the isotropy of Brownian motion to sample the \changed{exit location} distribution from maximal inscribed spheres. At each step, the random walk advances from the current position to a uniformly sampled point on the boundary of the largest sphere fully contained in the domain, resulting in rapid convergence without explicit time stepping.

Recent work has brought WoS-based solvers to the computer graphics and geometry processing communities. \citet{sawhney2020monte} introduced WoS for efficiently solving (screened) Poisson equations on complex geometries, along with spatial acceleration structures and gradient estimation techniques. Subsequent Walk-on-Stars (WoSt) algorithms generalized WoS to Neumann boundary conditions~\cite{sawhney2023walk} and to Robin boundary conditions~\cite{miller2024walkin}. In particular, \citet{miller2024walkin} employ the Brakhage-Werner trick together with an adaptive radius selection strategy to handle Robin boundary conditions within the WoSt framework. \citet{sugimoto2023practical} proposed an alternative approach based on Walk-on-Boundary, which handles a variety of boundary conditions but exhibits limited convergence in non-convex domains. \changed{\citet{huang2025geometric} extends the Robin version of walk-on-star algorithm from mesh setting to the closed implicit surfaces.}

Beyond direct solvers, a number of variants have been developed to extend the applicability of Monte Carlo PDE methods. 
Several works focus on making WoS solvers differentiable for inverse problems and optimization, including shape reconstruction and parameter estimation~\cite{miller2024differential, yu2024differential, yilmazer2024solving}.
\citet{yu2025robust} proposed a robust WoSt extension for the gradient estimation tasks.
These stochastic formulations have also been applied to geometric processing tasks such as cage-based deformation~\cite{de2024stochastic}. Extensions to more complex physical models include \changed{spatially varying coefficients}~\cite{sawhney2022grid}, Poisson Boolean media~\cite{miller2025solving}, and applications to fluid simulation~\cite{rioux2022monte, sugimoto2024velocity, jain2024neural}.

\changed{Most relevant to our work, \citet{bati2023coupling} proposed a coupled Monte Carlo solver for simulating infrared rendering, which models radiative heat transfer between surfaces, but resolves the radiative nonlinearity using a standard single-step linearization approach (Section~4.1, Equation 16 in their paper), corresponding to the simple linearization baseline in our evaluation. In contrast, our work focuses on enabling iterative nonlinear boundary updates within the Monte Carlo PDE framework to fully resolve the coupling. \citet[Chapter 7]{sabelfeld2016stochastic} described a stochastic fixed-point iteration method for nonlinear Poisson equations without practical numerical implementation and validation.}

\subsection{Variance Reduction}

A central challenge in Monte Carlo PDE solvers is variance reduction. Since solutions are computed via stochastic sampling and Monte Carlo integration, the resulting estimates can exhibit substantial noise, particularly in regions requiring many recursive evaluations. Consequently, a large body of work has focused on variance reduction strategies.

\citet{sawhney2020monte} studied classical variance reduction techniques including control variates and importance sampling in the context of WoS solvers.
Several following methods reduce variance by reusing information across nearby evaluation points. \citet{miller2023boundary} cache boundary values and normal derivatives to accelerate interior evaluations. \citet{bakbouk2023mean} exploit volumetric mean-value properties to spatially blend nearby samples, while \citet{qi2022bidirectional} draw inspiration from bidirectional path tracing to formulate a bidirectional WoS estimator. Neural caching approaches~\cite{li2023neural, li2024neural} replace explicit storage with learned predictors that estimate solution values. Other methods reconstruct local solution representations using harmonic expansions~\cite{zhou2025harmonic} or reuse samples from off-centered or overlapping spheres~\cite{czekanski2024walking, bao2025off}.

Despite these advances, existing variance reduction techniques primarily focus on solution estimates in the interior of the domain or leverage boundary information to accelerate interior queries. In contrast, variance reduction for \emph{on-boundary} solution estimates has received little attention. Boundary points do not admit enclosing interior spheres and therefore fall outside the scope of local sample reuse methods such as mean value caching, harmonic caching, or off-centered strategies. 

To address this gap, we introduce a heteroscedastic regression-based denoising method for boundary evaluations. By explicitly modeling spatially varying Monte Carlo uncertainty on the boundary, our approach reduces noise in radiative boundary estimates and integrates seamlessly with the fixed-point Monte Carlo framework proposed in this work.

\subsection{Heat Radiation Simulation}

Heat radiation plays a critical role in heat transfer modeling and has been extensively studied in both engineering and graphics applications. The radiative boundary condition introduces a nonlinear $T^4$ term arising from the Stefan--Boltzmann law, making the resulting systems highly nonlinear when discretized using \changed{finite difference (FD), finite element methods (FEM), or boundary element methods (BEM)}. Classical numerical solvers therefore rely on iterative schemes such as Picard iteration or Newton's method to resolve the nonlinearity~\cite[Appendix~2]{reddy2015introduction}.

In the context of FEM, Picard iteration is commonly adopted for its robustness and simplicity, and has been applied to radiative heat transfer problems in both mesh-based and meshless settings. For example, \citet{zhou2021heat} employ Picard iteration within a meshless weighted least-squares \changed{(MWLS)} framework. In computer graphics, \citet{freude2023precomputed} propose a radiative heat transfer solver based on photon tracing and precomputation, while \citet{freude2025inverse} extend this formulation to differentiable inverse problems. \changed{\citet{bialecki1981boundary} and \citet{blobner1999transient} explore BEM formulations for steady-state and transient radiative heat transfer.}

Despite these advances, existing approaches rely on spatial discretization, precomputation, or mesh-dependent representations. \changed{MWLS~\cite{zhou2021heat} is meshless, but the solution is represented as a linear combination of predefined basis functions, which constrains expressivity and requires basis design choices. In contrast, our boundary proxy uses direct Monte Carlo samples with local MLS reconstruction, avoiding both volumetric meshing and global basis construction while retaining geometric flexibility.} In contrast, the treatment of nonlinear radiative boundary conditions within Monte Carlo PDE solvers has received little investigation. Our work bridges this gap by integrating Picard-style fixed-point iteration directly into Monte Carlo estimators, enabling robust simulation of radiative heat transfer on complex geometries without discretization.

\subsection{Monte Carlo Estimators for Nonlinear Systems}

Monte Carlo estimation for nonlinear systems has attracted significant interest across multiple research communities, due to its fundamental role in nonlinear integral equations.

In the rendering community, nonlinear Monte Carlo estimators arise naturally in stylized rendering, where the shading is governed by nonlinear rendering equations~\cite{doi2021global, west2024stylized}. 
\changed{\citet{misso2022unbiased} proposed a Taylor-series debiasing method for nonlinear Monte Carlo estimators, which can be applied to a wide range of nonlinear rendering problems.}
Recently, \citet{tong2025practical} proposed a practical framework based on nonlinear path filtering and neural radiance caching to stabilize such estimators. 
Similarly, \citet{tinits2025nonlinear} introduced a nonlinear Noise2Noise formulation for training denoisers directly from noisy observations, addressing nonlinear statistical dependencies between inputs and targets.

\changed{Nonlinear Monte Carlo estimation also appears in volumetric rendering~\cite{novak2018monte}. In the radiative transfer equation, the transmittance along a ray is expressed as an exponential of an integral, which can be estimated using Monte Carlo methods. Null scattering techniques like Delta tracking~\cite{raab2006unbiased}, \changedd{the} next-flight estimator~\cite{kutz2017spectral}, and variants~\cite{georgiev2019integral} have been developed to efficiently handle the nonlinear transmittance term.}

In the machine learning community, estimating expectations of nonlinear recursive formulations has been studied in the context of nested Monte Carlo estimators~\cite{rainforth2018nesting}. From an applied mathematics perspective, multilevel Picard methods~\cite{beck2020nonlinear, hutzenthaler2021multilevel} provide a principled framework for solving nonlinear partial differential equations via stochastic fixed-point iteration. 

Despite these advances, directly applying existing nonlinear Monte Carlo estimators to Monte Carlo PDE solvers remains challenging. In radiative boundary value problems, the estimator is recursively branched at each reflection event, leading to rapidly growing variance and computational cost. Our work addresses this gap by introducing a fixed-point Monte Carlo formulation that is specifically designed for nonlinear radiative boundary conditions and remains stable under recursive boundary interactions.

\begin{table}
\centering
\caption{Notation used throughout the paper.}
\label{tab:notation}
\begin{tabular}{p{0.2\linewidth}p{0.75\linewidth}}
\toprule
\textbf{Symbol} & \textbf{Description} \\
\midrule
\multicolumn{2}{l}{\textbf{Domain and Boundary}} \\
$\Omega$ & Bounded domain in $\mathbb{R}^3$ \\
$\partial\Omega_D$ & Dirichlet boundary region \\
$\partial\Omega_N$ & Neumann boundary region \\
$\partial\Omega_R$ & Robin/radiative boundary region \\
$\vec{n}$ & Outward unit normal vector on boundary \\
\midrule
\multicolumn{2}{l}{\textbf{PDE and Physics}} \\
$T(x)$ or $u(x)$ & Temperature field (solution to the PDE) \\
$k_0$ & Thermal conductivity \\
$\sigma$ & Stefan-Boltzmann constant ($5.67\times10^{-8}$ W/m$^2$K$^4$) \\
$\varepsilon$ & Surface emissivity \\
$T_0, T_f, T_s$ & Ambient temperatures in various contexts \\
$\mu(x)$ & Robin coefficient \\
\midrule
\multicolumn{2}{l}{\textbf{Monte Carlo Solver}} \\
$\text{St}(x,R)$ & Star-shaped region used in Walk-on-Stars \\
$B(x,R)$ & Ball of radius $R$ centered at $x$ \\
$P^\Omega(x,z)$ & Poisson kernel \changed{of region $\Omega$} \\
$G^\Omega(x,z)$ & Green's function \changed{of region $\Omega$} \\
$\rho_\mu(x,z)$ & Reflectance coefficient \\
\bottomrule
\end{tabular}
\end{table}

\section{Background}

Before introducing our method, we \changed{briefly review} the physical model of heat transfer and Monte Carlo PDE solvers, especially the Walk-on-Spheres family. The key symbols used throughout the paper are summarized in Table~\ref{tab:notation}.

\subsection{Physical Model of Heat Transfer}

We consider steady-state heat transfer in a bounded domain $\Omega \subset \mathbb{R}^3$ with mixed boundary conditions, including radiative heat exchange. The temperature field $T(x)$ satisfies the following boundary value problem~\cite{modest2021radiative}:
\begin{subequations}\label{equ:heat_bvp}
\begin{align}
    &k_0 \Delta T(x) = 0, && x \in \Omega, \label{equ:heat_bvp:a} \\[1ex]
    &T(x) = T_0(x), && x \in \partial\Omega_{\changed{D}}, \label{equ:heat_bvp:b} \\[1ex]
    &\frac{\partial T}{\partial \vec n}(x) = \frac{q_0(x)}{k_0}, && x \in \partial\Omega_{\changed{N}}, \label{equ:heat_bvp:c} \\[1ex]
    &\frac{\partial T}{\partial \vec n}(x)
    + \frac{h}{k_0}\bigl[T(x)-T_f(x)\bigr] \nonumber \\
    &\phantom{\frac{\partial T}{\partial \vec n}(x)}
    + \frac{\sigma\epsilon}{k_0}\bigl[T(x)^4 - T_s(x)^4\bigr] = 0,
    && x \in \partial\Omega_{\changed{R}}. \label{equ:heat_bvp:d}
\end{align}
\end{subequations}

Equation~\eqref{equ:heat_bvp:a} models heat conduction inside the domain, governed by the Laplace equation with constant thermal conductivity $k_0$. The boundary $\partial\Omega_{\changed{D}}$ is subject to prescribed temperature values (Dirichlet condition), while $\partial\Omega_{\changed{N}}$ enforces a prescribed heat flux (pure Neumann condition), where $q_0(x)$ denotes the outward heat flux density.

The boundary $\partial\Omega_{\changed{R}}$ models combined heat convection and radiation. The linear term
\(
\frac{h}{k_0}[T(x)-T_f(x)]
\)
represents convective heat exchange with a surrounding fluid of temperature $T_f(x)$, where $h$ is the convection coefficient. The nonlinear term
\(
\frac{\sigma\epsilon}{k_0}[T(x)^4 - T_s(x)^4]
\)
models radiative heat transfer according to the Stefan--Boltzmann law, with emissivity $\epsilon$, Stefan--Boltzmann constant $\sigma$, and ambient radiative temperature $T_s(x)$.

In many scenarios, convection dominates and the radiative term can be neglected, reducing~\eqref{equ:heat_bvp:d} to a Robin boundary condition:
\begin{equation}
    \frac{\partial T}{\partial \vec{n}} + \frac{h}{k_0}\bigl(T(x) - T_f(x)\bigr) = 0,
    \qquad x \in \partial\Omega_{\changed{R}}.
\end{equation}
However, in vacuum or near-vacuum environments, such as spacecraft, satellites, or asteroids, radiation becomes the primary heat transfer mechanism, making the nonlinear boundary condition unavoidable and numerically challenging.

\subsection{Monte Carlo PDE Solvers}

Monte Carlo PDE solvers provide a mesh-free, geometry-agnostic framework for solving elliptic boundary value problems. We briefly review the Walk-on-Spheres method and its extensions to different boundary conditions from the perspective of boundary integral equations.

Consider a general Poisson problem with mixed boundary conditions:
\begin{equation}
\label{equ:bvp_robin}
\left\{
\begin{aligned}
    u(x) &= g(x), && x \in \partial \Omega_D, \\
    \frac{\partial u}{\partial \vec{n}}(x) + \mu(x)\, u(x) &= h(x), && x \in \partial \Omega_R, \\
    \Delta u(x) &= -f(x), && x \in \Omega.
\end{aligned}
\right.
\end{equation}

\paragraph{Boundary Integral Equation.}
Assuming sufficient regularity, the solution admits a boundary integral equation. For any region $A \subset \Omega$ and auxiliary domain $C \subset \mathbb{R}^n$, we have
\begin{equation}
\label{equ:bie}    
    \begin{aligned}
        \alpha(x)\, u(x)
        &=
        \int_{\partial A} \Bigl[
            P^C(x,z)\, u(z)
            - G^C(x,z)\, \frac{\partial u}{\partial \vec{n}}(z)
            \Bigr] \,\mathrm{d}S_z \\
            &+
            \int_A G^C(x,y)\, f(y)\,\mathrm{d}y,
        \end{aligned}
\end{equation}
where $\alpha(x)=1$ for $x\in A$ and $\alpha(x)={1}/{2}$ for $x\in\partial A$. When $C$ is chosen as a ball, the Poisson kernel $P^C$ and Green's function $G^C$ admit closed-form expressions. \changed{It is also worth noting} that $G^{B(x, R)}(x, y)$ vanishes when $y \in \partial B(x, R)$.

\subsubsection{Dirichlet Boundary: Walk-on-Spheres}

For Dirichlet problems, the Walk-on-Spheres algorithm~\cite{muller1956some, sawhney2020monte} chooses $A=C=B(x,R)$, where $R=d(x,\partial\Omega)$ \changed{is the distance between $x$ and $\partial\Omega$}. Since $G^C$ vanishes on $\partial C$, equation~\eqref{equ:bie} reduces to
\begin{equation}
    u(x)
    =
    \int_{\partial B(x,R)} P^B(x,z)\, u(z)\,\mathrm{d}S_z
    +
    \int_{B(x,R)} G^B(x,y)\, f(y)\,\mathrm{d}y.
\end{equation}
This formulation naturally yields a recursive Monte Carlo estimator by sampling points on the sphere boundary and within the ball. The walk terminates when the distance to $\partial\Omega$ falls below a prescribed threshold.

\subsubsection{Neumann Boundary: \changed{Walk-on-Stars}}

To handle Neumann boundary conditions, \citet{sawhney2023walk} proposed the Walk-on-Stars algorithm. Instead of a full ball, a star-shaped region
\(
\changed{\mathrm{St}}(x,R) = \Omega \cap B(x,R)
\)
is constructed, where $R$ is chosen as the closest silhouette distance to the Neumann boundary. We denote $\mathrm{St}_N = \partial\Omega \cap St$ and $\mathrm{St}_B = \partial St \setminus \mathrm{St}_N$, \changed{representing the \changedd{boundaries lying} on the sphere} $\partial B$ \changed{and on the domain boundary} $\partial\Omega$. Since $G^B$ vanishes on $\partial B(x,R)$, the corresponding boundary integral equation becomes
\begin{equation}
    \begin{aligned}
        \label{equ:bie_neumann}
        \alpha(x) u(x)
        &=
        \int_{\partial \mathrm{St}} P^B(x,z)\, u(z)\,\mathrm{d}S_z
        -
        \int_{\partial \mathrm{St}_N} G^B(x,z)\, \frac{\partial u}{\partial \vec{n}}(z)\,\mathrm{d}S_z \\
        &\quad +
        \int_{\mathrm{St}} G^B(x,y)\, f(y)\,\mathrm{d}y,
    \end{aligned}
\end{equation}
where the normal derivative $\sfrac{\partial u}{\partial \vec{n}}$ on $\partial \mathrm{St}_N$ is provided by the Neumann boundary condition. \citet{sawhney2023walk} described sampling methods for each integral term, yielding an unbiased Monte Carlo estimator for Neumann problems.

\subsubsection{Robin Boundary: \changed{Walk-on-Stars} with Reflectance Term}

For Robin boundary conditions, \citet{miller2024walkin} apply the Brakhage--Werner trick to eliminate the normal derivative. Substituting
$
\sfrac{\partial u}{\partial \vec{n}} + \mu u = h
$ into~\eqref{equ:bie_neumann} yields
\begin{equation}
    \begin{aligned}
        \alpha(x) u(x)
        &=
        \int_{\partial \mathrm{St}(x,R)} \rho_\mu(x,z)\, P^B(x,z)\, u(z)\,\mathrm{d}S_z \\
& +
\int_{\partial \mathrm{St}_R(x,R)} G^B(x,z)\, h(z)\,\mathrm{d}S_z \\
&+
\int_{\mathrm{St}(x,R)} G^B(x,y)\, f(y)\,\mathrm{d}y,
\end{aligned}
\end{equation}
where $\partial \mathrm{St}_R = \partial\Omega \cap \mathrm{St}$ and $\partial \mathrm{St}_B = \partial \mathrm{St} \setminus \partial \mathrm{St}_R$.

The reflectance coefficient is defined as
\begin{equation}
\label{equ:rho_nonlinear}
\rho_\mu(x,z) =
\begin{cases}
1 - \dfrac{G^B(x,z)}{P^B(x,z)}\, \mu(z),
& z \in \partial \mathrm{St}_R, \\[6pt]
1, & z \in \partial \mathrm{St}_B.
\end{cases}
\end{equation}
To ensure estimator stability, the star radius $R$ is adaptively reduced such that
\(
\rho_\mu(x,z) \in [0,1]
\)
for all $z \in \partial \mathrm{St}(x,R)$.
\citet{miller2024walkin} implemented this constraint on triangle meshes, embedding the value of $\rho$ into the spatial acceleration data structure. To build the spatial structure ensuring $\rho_\mu(x,z) \in [0,1]$ for all $z \in \partial \mathrm{St}(x,R)$, users should provide the maximum and minimum values of $\mu$ in each triangle.

\begin{figure*}[ht]
    \centering
    \includegraphics[width=\linewidth]{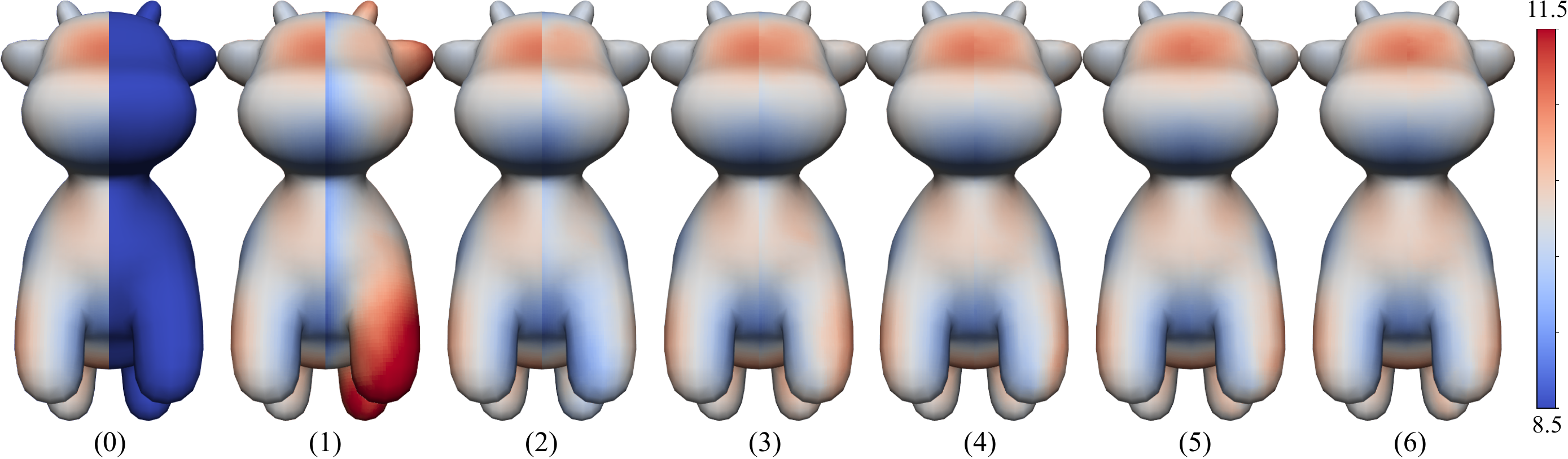}
    \caption{Visualization of the fixed-point iteration process for the radiative boundary condition. The left half of each image shows the solution on the Dirichlet boundary, while the right half shows the corresponding solution on the radiative boundary. The ground-truth solution is symmetric across the two halves. Panel (0) shows the initial proxy, and panels (1)--(6) show the boundary solution after the $n$-th fixed-point iteration, illustrating progressive convergence toward the ground truth.}
    \label{fig:iteration}
\end{figure*}

\section{Method}
\label{sec:method}

We propose an iterative Monte Carlo algorithm to solve heat transfer problems with radiative boundary conditions on complex geometries. Our method addresses the core challenge posed by the nonlinear $T^4$ term, which prevents direct application of classic Monte Carlo PDE solvers. By integrating the fixed-point iteration framework with the Walk-on-Stars (WoSt) estimator, we linearize the boundary at each step and update the solution until \changed{it reaches} the fixed point, which is the solution of the original boundary value problem. Algorithm~\ref{alg:main} summarizes the overall procedure.

For notational convenience, we consider the following boundary value problem in this section:
\begin{equation}
\label{equ:bvp_radiative}
\left\{
\begin{aligned}
    u(x) &= g(x), && x \in \partial \Omega_D, \\
    \frac{\partial u}{\partial \vec{n}}(x) + \mu(x)\, u^4(x) &= h(x), && x \in \partial \Omega_R, \\
    \Delta u(x) &= -f(x), && x \in \Omega.
\end{aligned}
\right.
\end{equation}
where $\partial \Omega_D$ is the Dirichlet boundary and $\partial \Omega_R$ is the radiative boundary. The coefficient $\mu(x)$ and \changed{solution} $u(x)$ are both non-negative. Although we focus on pure radiative boundaries for clarity, our method extends naturally to mixed convection--radiation conditions by incorporating the linear convection term into the Robin coefficient $\mu(x)$. Extensions to spatially varying conductivity and screened Poisson equations are also straightforward~\cite{sawhney2022grid,miller2024walkin}.

\begin{algorithm}
    \caption{Solving Radiative BVP with Picard Iteration}
    \label{alg:main}
    \begin{algorithmic}[1]
    \State \textbf{Input:} Geometry specified by boundary $\partial \Omega$, \changed{split into the} absorbing boundary $\partial \Omega_D$ and radiative boundary $\partial \Omega_R$. Number of iterations $N$. Relaxation coefficient $\alpha$.
    \State \textbf{Output:} Solution on radiative boundary $\partial \Omega_R$.
    
    \State $u_0 \gets InitialGuess()$
    
    \For{$n = 1$ to $N$}
        \State $\mathrm{BuildRobinAccelerationStructure}(u_{n - 1})$
        \State $u_n \gets ~$ Solve the BVP with proxy solution $u_{n - 1}$ on $\partial \Omega_R$.
        \State Relaxation: $u_n \gets \alpha u_n + (1 - \alpha) u_{n - 1}$
    \EndFor
    
    \State (Optional) $u_n \gets RegressionDenoiser(\partial \Omega_{\changedd{R}}, u_n)$
    \State \textbf{Return} $u_N$
    \end{algorithmic}
\end{algorithm}

\subsection{Fixed-Point Iteration Framework}
\label{sec:iteration}

The primary difficulty in handling radiative boundary conditions lies in the nonlinear term $u^4$, which prevents the direct application of classical one-sample Monte Carlo estimation. To address this issue, we adopt a Picard-style fixed-point strategy in which the nonlinearity is handled through successive linearized boundary value problems.

A typical approach \changed{for handling} the nonlinear term is linearization through a proxy solution. Let $u_0$ denote an initial proxy solution on the radiative boundary $\partial\Omega_R$. The nonlinear term $u^4$ is linearized by freezing the cubic coefficient. Specifically, the radiative contribution $u^4$ is approximated as
\begin{equation}
u^4 \approx u_0^3 u,
\end{equation}
where $u_0$ is treated as known. Substituting this approximation into the radiative boundary condition yields the following Robin-type condition:
\begin{equation}
\label{equ:linearized}
\frac{\partial u}{\partial \vec{n}} (x)
+ \mu(x) u_0^3(x) u(x)
= h(x).
\end{equation}
Under this approximation, the original nonlinear radiative boundary condition is reduced to a linear Robin boundary condition, which can be handled using the modified Walk-on-Stars estimator~\cite{miller2024walkin}.

This coefficient-freezing strategy works well when the proxy solution $u_0$ is close to the true solution \changed{and} solution variations along the boundary are moderate.
However, if the initial proxy is inaccurate, a single linearization step may introduce systematic bias. 
For example, if the proxy solution is underestimated, the resulting solution is greater than the correct solution, which is formally described by the following comparison theorem.

\begin{lemma}[Comparison Theorem]
\label{thm:comparison}
Let $u$ solve the nonlinear radiative boundary value problem \eqref{equ:bvp_radiative}, and let $u_0$ solve the linearized problem \eqref{equ:linearized} with proxy $p(x) \geq 0$.  
If
\begin{equation}
p(x) \le u(x) \quad \text{for all } x \in \partial\Omega_R,
\end{equation}
then
\begin{equation}
u_0(x) \ge u(x) \quad \text{for all } x \in \Omega.
\end{equation}
\end{lemma}

\changed{This comparison relation follows the standard maximum principle argument~\cite{evans2022partial}, and we include a short proof in Appendix~\ref{sec:proof1}.}

\begin{figure}[h]
    \centering
    \includegraphics[width=0.8\linewidth]{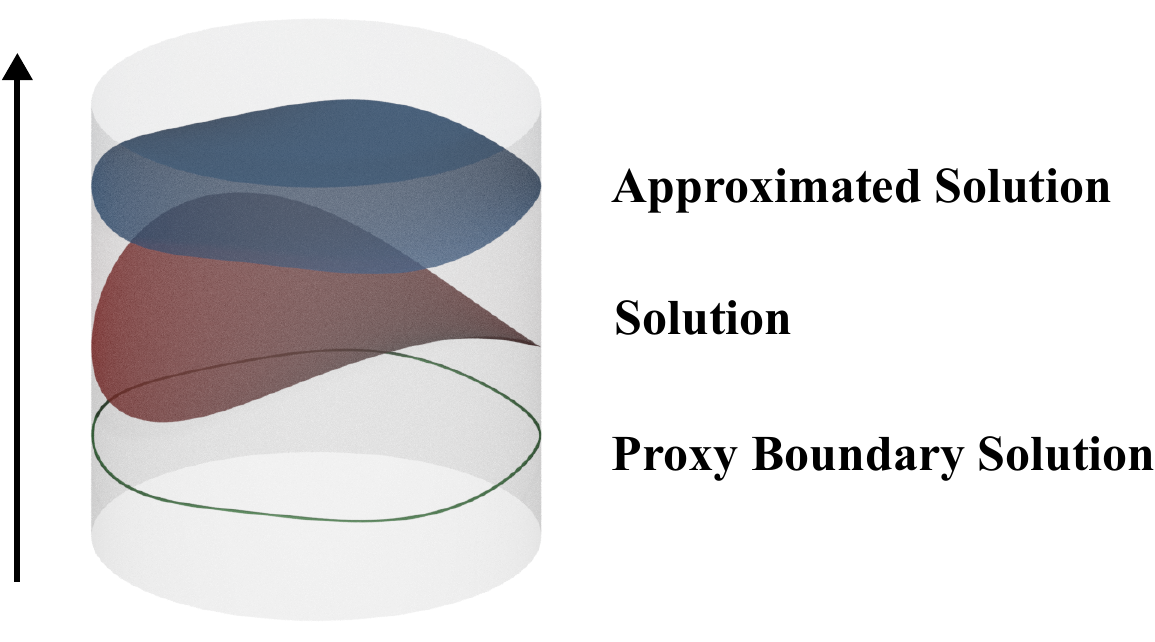}
    \caption{Illustrative figure of the comparison theorem \ref{thm:comparison}. The green curve is the image of an underestimated radiative boundary value. Blue surface is the image of the solution of the linearized boundary value problem, which is larger than the correct solution (shown in red).}
    \label{fig:comparison_thm}
\end{figure}

Figure \ref{fig:comparison_thm} intuitively presents this property, and we provide the proof in Appendix~\ref{sec:proof1}. An overestimated initial guess will lead to an underestimated solution. To mitigate this issue, we embed the linearization into a fixed-point iteration framework. We formalize this process by defining an operator $\mathcal{S}(u_p, f, g, h)$ that maps a proxy solution $u_p$ on $\partial\Omega_R$ to the resulting solution on $\partial\Omega_R$, obtained by solving the corresponding linearized boundary value problem with fixed coefficients. The solution of the original nonlinear problem is a fixed point of this operator:
\begin{equation}
u = \mathcal{S}(u, f, g, h).
\end{equation}
Accordingly, we apply the Picard iteration
\begin{equation}
u_{n+1} = \mathcal{S}(u_n, f, g, h),
\end{equation}
using the solution from the previous iteration as the proxy for coefficient evaluation.

Substituting this formulation into the boundary integral equation associated with a star-shaped domain $St(x, R)$ yields
\begin{equation}
\begin{aligned}
\alpha(x) u_{n+1}(x)
&=
\int_{\partial \mathrm{St}(x,R)} \rho_\mu(x,z) P^B(x,z) u_{n+1}(z)\,\mathrm{d}z \\
&+ \int_{\partial \mathrm{St}_R(x,R)} G^B(x,z) h(z)\,\mathrm{d}z \\
& + \int_{\mathrm{St}(x,R)} G^B(x,y) f(y)\,\mathrm{d}y,
\end{aligned}
\end{equation}
where the reflectance term $\rho_\mu$ is given by
\begin{equation}
\label{equ:rho}
\rho_\mu(x,z) =
\begin{cases}
1 - \dfrac{G^B(x,z)}{P^B(x,z)} \mu(z) u_n^3(z),
& z \in \partial \mathrm{St}_R, \\
1, & z \in \partial \mathrm{St}_B.
\end{cases}
\end{equation}

A direct Monte Carlo evaluation of $u_{n+1}$ is challenging because $\rho_\mu$ depends on the unknown solution through $u_n$. Naively estimating this dependence would require recursively sampling multiple coupled quantities, leading to branching estimators whose computational cost grows exponentially with recursion depth. To avoid this issue, we adopt a caching-based strategy: at each iteration, the solution on the radiative boundary is estimated once, stored as a proxy, and reused to evaluate $\rho_\mu$ in subsequent estimations. This eliminates branching and keeps the overall computational complexity linear in the number of Monte Carlo samples.

\begin{figure}[h]
    \centering
    \includegraphics[width=\linewidth]{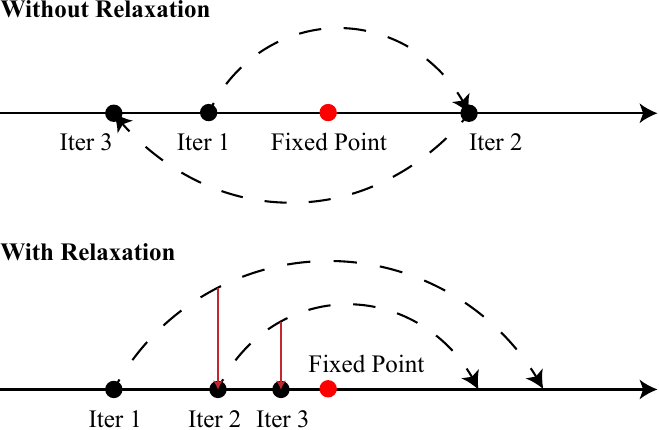}
    \caption{Illustration of relaxation in fixed-point iterations. Without relaxation, successive iterates may oscillate and fail to converge. Introducing a relaxation factor damps these oscillations and significantly improves convergence stability.}
    \label{fig:relaxation}
\end{figure}

\paragraph{Bias, convergence and stability.}

The Picard iteration may exhibit oscillatory behavior or fail to converge due to the strong nonlinearity of the boundary condition, as the induced fixed-point map is not guaranteed to be contractive. To improve stability, we apply under-relaxation at each iteration:
\begin{equation}
    u_n \gets \alpha u_n + (1 - \alpha) u_{n-1}, \quad \alpha \in (0,1).
\end{equation}
Empirically, relaxation significantly improves convergence and suppresses oscillations. Figure~\ref{fig:relaxation} illustrates the stabilizing effect of relaxation, while Figure~\ref{fig:iteration} shows the iterative convergence behavior under the configuration described in Section~\ref{sec:synthetic}. Starting from an underestimated initial solution, the method progressively converges to the ground truth.
\changed{Nevertheless, this behavior is empirical: the fixed-point iteration is not theoretically guaranteed to converge globally, and oscillatory or divergent behavior can occur for unfavorable parameter settings (see Section~\ref{sec:failure_cases} and Figure~\ref{fig:failure}).}

Our method evaluates the nonlinear radiative term using a cached proxy solution on the boundary. Since radiative emission depends on a nonlinear function, in general
\begin{equation}
\mathbb{E}[u^3] \neq \mathbb{E}[u]^3,
\end{equation}
and therefore directly applying the nonlinear operator to the cached mean introduces bias into the Picard iteration process.

In practice, this bias is relatively small compared \changed{\changedd{to} the systematic bias caused} by \changedd{an} \changed{improper} proxy solution\changedd{,} and \changedd{it} can be reduced by increasing the number of samples. \changed{Empirically, we observe stable convergence and accurate solutions across tested scenarios, while still acknowledging the failure modes discussed in Section~\ref{sec:failure_cases}.}

\subsection{Boundary Function Representation}
\label{sec:boundary_representation}

Within the fixed-point iteration framework, each iteration requires evaluating and updating an intermediate solution field $u_n$ defined on the radiative boundary $\partial\Omega_R$. This necessitates a flexible boundary function representation that is compatible with Monte Carlo estimation and does not rely on mesh-dependent discretization.

A common choice for surface function representation is a per-vertex or per-face discretization. However, both approaches fundamentally tie the resolution of the function space to the underlying mesh. In particular, when the mesh contains large primitives, high-frequency variations induced by radiative effects cannot be adequately captured without mesh refinement or \changed{remeshing} (Figure \ref{fig:wireframe}). Introducing adaptive surface refinement solely for boundary function representation contradicts the processing-free design principle of Monte Carlo PDE solvers, which aims to avoid geometry-dependent preprocessing.

\begin{figure}[h]
    \centering
    \includegraphics[width=\linewidth]{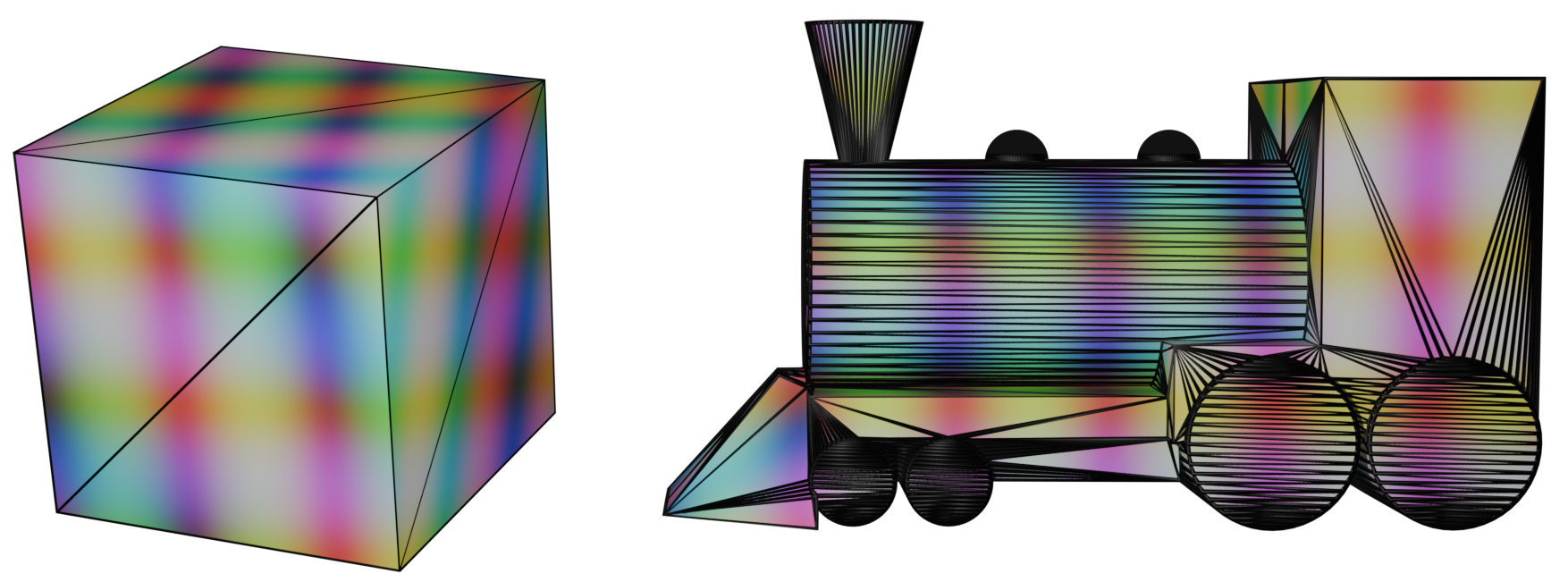}
    \caption{Failure cases of per-vertex or per-face function representations. Even for simple geometry such as a cube (a), each face is represented by only two triangles, which is insufficient to resolve spatially varying boundary quantities. This limitation becomes more severe for complex meshes with nonuniform triangulation, such as the train model (b).}
    \label{fig:wireframe}
\end{figure}

To decouple boundary function fidelity from mesh resolution, we adopt a point-cloud-based representation. At each iteration, a set of $N$ points $\{P_i\}_{i=1}^N$ is uniformly sampled from the radiative boundary $\partial\Omega_R$. For each sample point $P_i$, the value $u_n(P_i)$ is estimated independently using the Walk-on-Stars estimator applied to the current linearized boundary value problem.

To evaluate the proxy solution at arbitrary boundary locations, we reconstruct a continuous boundary function from these scattered samples using a moving least squares (MLS) approximation. Given an evaluation point $x \in \partial\Omega_R$, the proxy is locally approximated by a linear model
\begin{equation}
u_n(x) \approx a + b^\top (x - x_0),
\end{equation}
where the coefficients $a \in \mathbb{R}$ and $b \in \mathbb{R}^3$ are determined by minimizing the weighted least-squares objective
\begin{equation}
E(a,b) = \sum_{i} w(\|x - P_i\|)
\left(a + b^\top (P_i - x) - u_n(P_i)\right)^2.
\end{equation}
Only sample points within a fixed-radius neighborhood of $x$ are considered. The weight function is chosen as
\begin{equation}
w(r) = \exp\!\left(-\frac{r^2}{h^2}\right),
\end{equation}
with fixed bandwidth parameters used throughout all experiments. \changed{Pseudocode is} listed in Appendix \ref{sec:appendix_mls}. \changed{For efficient neighborhood queries, boundary samples are indexed with a KD-tree rebuilt each iteration; in practice, this step is inexpensive relative to Monte Carlo sampling and is not a runtime bottleneck.}

\paragraph{Updating Robin coefficients.}
After each iteration, the updated proxy solution $T_n$ induces new Robin coefficients in the linearized boundary condition. For the Walk-on-Stars estimator, it is sufficient to provide conservative lower and upper bounds on the Robin coefficient over each boundary primitive. For each triangle on $\partial\Omega_R$, we collect all proxy samples associated with that triangle, as well as the vertex values, and compute corresponding bounds on the coefficient
\begin{equation}
\mu(x) u_n^3(x).
\end{equation}
The BVH in the Walk-on-Stars~\cite{miller2024walkin} is re-initialized \changed{using} the minimum and the maximum of these coefficients. Due to the MLS approximation, the evaluated \changed{coefficient} value in the triangle may be larger than the estimated bound. \changed{In all experiments, we expand each bound by a fixed relative margin to ensure feasibility of the WoSt reflectance constraint.} \changedd{Thus, the} minima and maxima are further expanded by a margin to ensure the stability of the estimator.

\subsection{On-Boundary Denoising via Heteroscedastic Regression}
\label{sec:denoising}

Monte Carlo PDE solvers inherently produce stochastic noise whose variance depends on geometry, boundary conditions, and estimator configuration. Existing variance reduction techniques primarily focus on reducing noise in interior evaluations by exploiting spatial correlations. In contrast, our method requires repeatedly estimating the solution trace on the boundary, where variance reduction has received comparatively little attention.

To stabilize the fixed-point iteration and reduce boundary noise, we introduce an on-boundary denoising strategy based on heteroscedastic regression. For each evaluation point $x \in \partial\Omega_R$, the Monte Carlo estimator produces a noisy estimate
\begin{equation}
\hat{u}(x) = u(x) + \epsilon(x),
\end{equation}
where the noise $\epsilon(x)$ is modeled as a zero-mean Gaussian random variable with spatially varying variance,
\begin{equation}
\epsilon(x) \sim \mathcal{N}(0,\sigma^2(x)).
\end{equation}
Unlike standard regression settings with homoscedastic noise, the variance $\sigma^2(x)$ in Monte Carlo estimators is spatially varying, \changed{because estimates at different boundary positions arise from different} boundary integral representations~\eqref{equ:bie}.

\begin{figure}
    \centering
    \includegraphics[width=\linewidth]{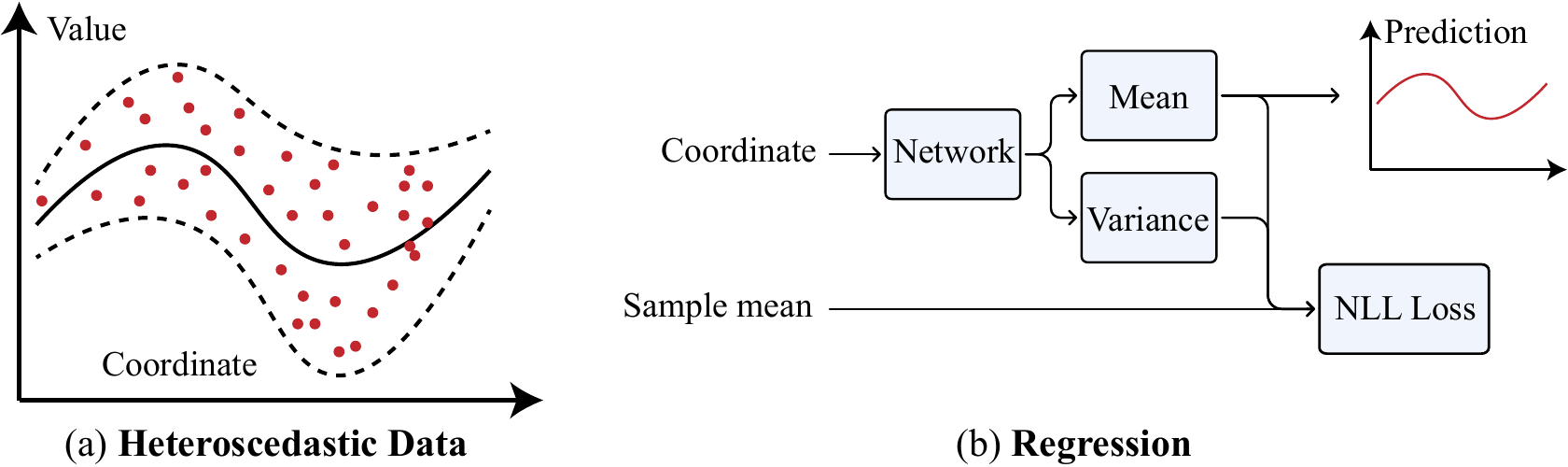}
    \caption{Illustration of heteroscedastic regression for boundary denoising. Left: noisy and spatially varying Monte Carlo estimates on the boundary. Right: a neural network jointly learns the mean and variance fields, explicitly modeling heteroscedastic uncertainty in the boundary solution.}
    \label{fig:regression}
\end{figure}

We therefore jointly regress both the mean solution field $u(x)$ and the local standard deviation $\sigma(x)$ using compact neural networks. The networks are trained by minimizing the negative log-likelihood~\cite{nix1994estimating}
\begin{equation}
\mathcal{L}(\hat{u},\sigma)
=
\mathbb{E}_{x \sim \partial\Omega}
\left[
\frac{(\hat{u}(x) - u(x))^2}{2\sigma^2(x)}
+ \frac{1}{2}\log \sigma^2(x)
\right].
\end{equation}
This formulation explicitly accounts for spatially varying uncertainty and prevents overfitting to high-variance samples.

In our implementation, we employ SIREN~\cite{sitzmann2020implicit} networks due to their strong empirical performance in representing smooth spatial signals. Figure \ref{fig:regression} illustrates this process and our network architecture. However, the denoising framework is model-agnostic, and alternative representations such as Fourier feature networks~\cite{tancik2020fourier} can be used interchangeably without modifying the underlying formulation. We also incorporate the $\beta$-NLL ~\cite{seitzer2022pitfalls} to enhance the stability of gradient-based optimization.

\changed{The denoiser is not pre-trained. For each denoising pass, we train a scene-specific regression model from a random initialization using the current Monte Carlo boundary samples. In the experiments reported here, denoising is applied once to the final boundary estimate after the fixed-point iterations.}

% The denoiser is applied after each Monte Carlo estimation stage and before the proxy solution is used to update boundary coefficients. Since the denoiser is trained on Monte Carlo estimates and does not alter the expectation of the estimator in the infinite-sample limit, the fixed-point iteration remains consistent in expectation. Empirically, denoising substantially reduces iteration-to-iteration oscillations and accelerates convergence of the nonlinear solver.

% \subsection{Other Non-linear Boundary Conditions}

% \section{Implementation}

\section{Results}

In this section, we evaluate the proposed method \changed{across a variety of setups} and compare \changed{it} with the naive linearization method and \changed{a well-established} FEM solver.

\subsection{Implementation}

We implement the proposed Monte Carlo solver in C++, with geometric queries and closest-distance evaluations provided by FCPW~\cite{FCPW}. Boundary denoising is performed using a heteroscedastic regression model implemented in PyTorch~\cite{pytorch}. Following prior work, we adopt a SIREN architecture~\cite{sitzmann2020implicit} for its ability to represent smooth functions with high-frequency detail, and optimize the network using the $\beta$-NLL loss~\cite{seitzer2022pitfalls} with $\beta = 0.5$. The neural network is trained using the Adam optimizer~\cite{Kingma2014AdamAM} for 3000 epochs. The learning rate is initialized at $10^{-3}$ and decayed to $10^{-4}$ using a cosine annealing schedule. We run the Monte Carlo solver on an Intel Core i7-10700K CPU with thread-level parallelization, and train the heteroscedastic regression model on an NVIDIA A100 GPU with 40\,GB memory.

In heat-related equations such as \eqref{equ:heat_bvp:d}, the nonlinear $T^4$ term can become numerically large and the Stefan--Boltzmann constant $\sigma$ is very small, which may lead to instability. To improve numerical stability, we rescale the temperature field by a factor $c$, so the radiative term becomes
\begin{equation}
\frac{\sigma\epsilon c^4}{k_0} T(x)^4.
\end{equation}
We choose $c$ such that the rescaled coefficient of the $T^4$ term is $10^{-3}$ in all experiments.

\subsection{Evaluation with Synthetic Functions}
\label{sec:synthetic}

\begin{figure}
    \centering
    \includegraphics[width=\linewidth]{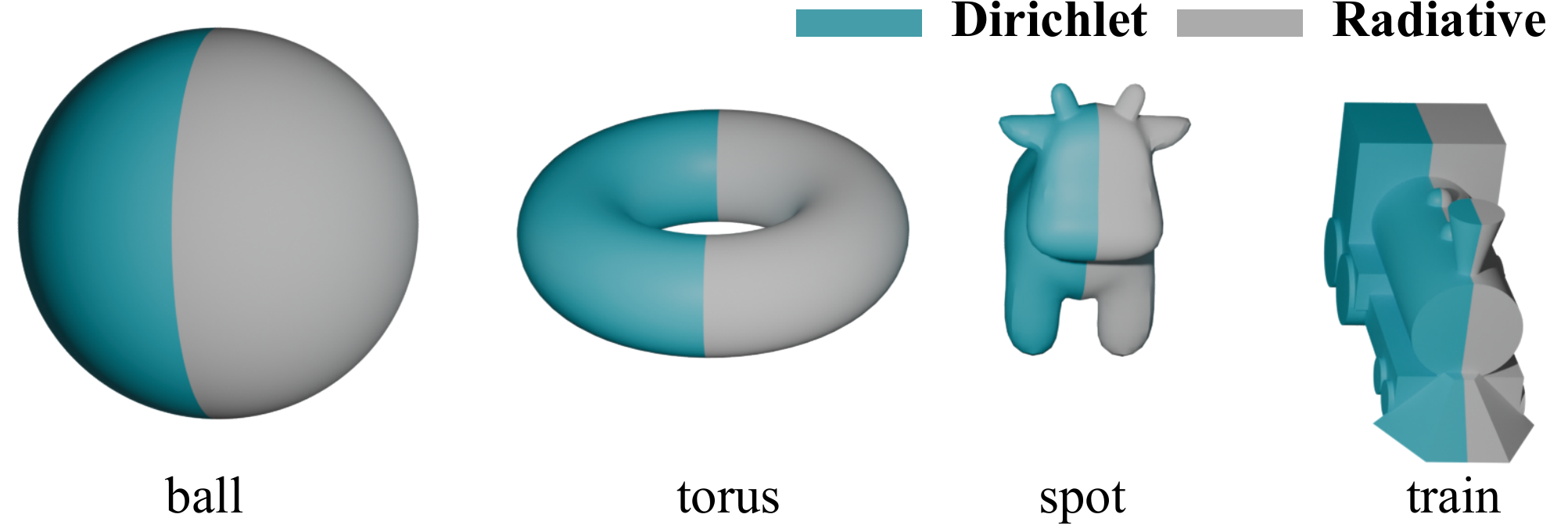}
    \caption{Models for evaluation. Each model is bisected into two parts, one for the Dirichlet boundary and the other for the \changed{radiative} boundary.}
    \label{fig:geometry}
\end{figure}

\begin{figure*}
    \centering
    \includegraphics[width=\linewidth]{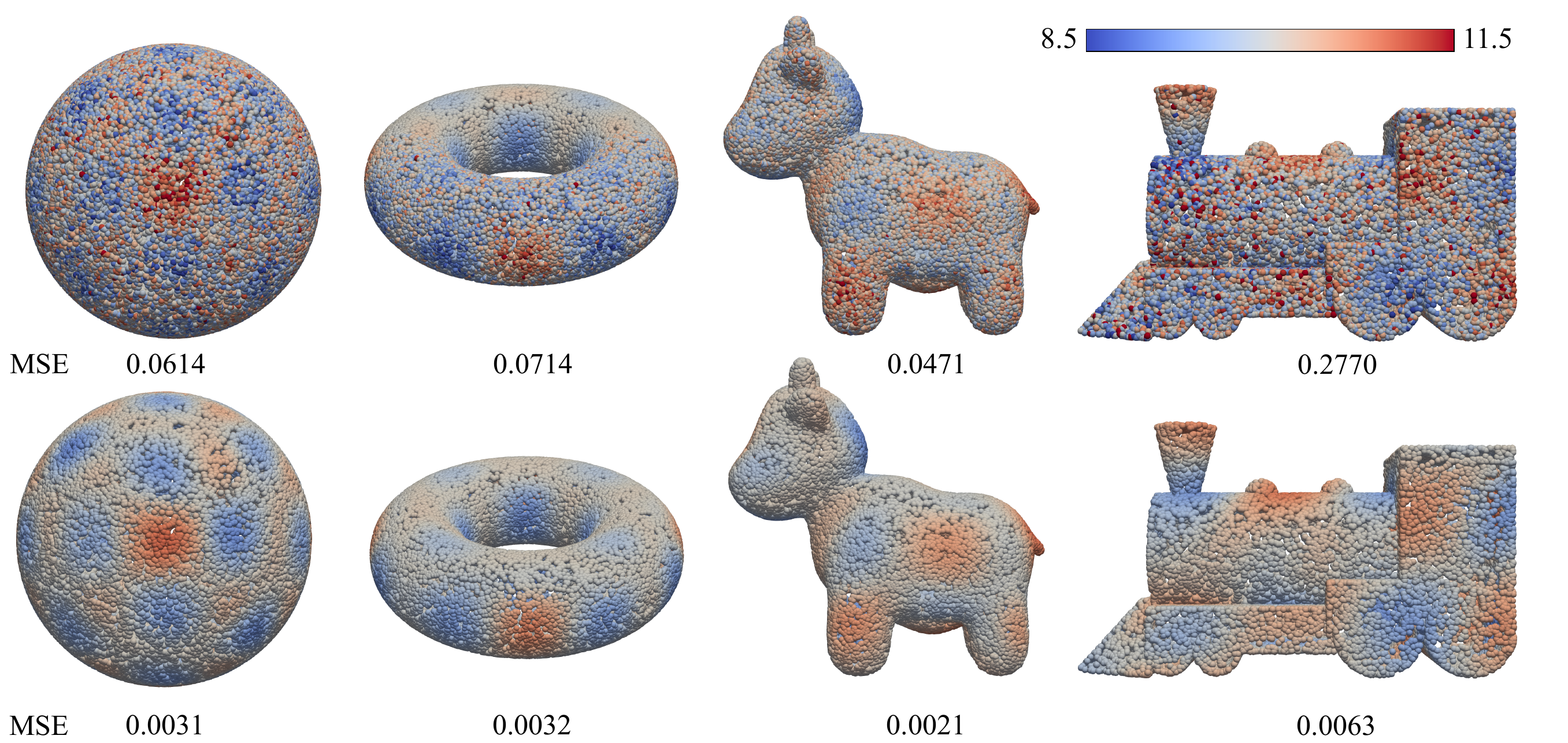}
    \caption{Effect of boundary denoising on Monte Carlo estimates. \changed{Here we are visualizing the MLS point cloud.} The first row shows the raw boundary estimates obtained from the Monte Carlo solver before denoising, while the second row shows the corresponding results after applying the proposed heteroscedastic regression-based denoising method. From the MSE value, we can see that the denoising process significantly reduces the variance of the estimates, leading to a much smoother and more accurate boundary solution.}
    \label{fig:denoise}
\end{figure*}

We evaluate the proposed fixed-point Monte Carlo solver using synthetic solutions with analytically constructed boundary conditions. This setting enables quantitative error measurement and controlled assessment of convergence behavior under nonlinear radiative boundary conditions.

The ground-truth solution is defined as
\begin{equation}
u(x) = 10 + \cos(\omega x)\cos(\omega y)\cos(\omega z),
\end{equation}
where $\omega$ controls the spatial frequency of the solution and is fixed across all experiments. The domain boundaries (Figure~\ref{fig:geometry}) are partitioned into two disjoint regions: a Dirichlet boundary on the left half and a radiative boundary on the right half. On the radiative boundary, the solution $u$ satisfies
\begin{equation}
\frac{\partial u}{\partial \vec{n}} + \gamma u^4 = f(x),
\end{equation}
where $\gamma = 10^{-3}$ and the source term $f(x)$ is analytically computed such that the ground-truth solution $u(x)$ satisfies the boundary condition exactly.

To evaluate robustness with respect to initialization, we intentionally choose a constant initial proxy solution of $8$ on the radiative boundary, which deviates significantly from the true boundary solution. This setup tests the ability of the fixed-point iteration to recover from an inaccurate initial guess.

At each iteration, we uniformly sample $10^4$ points on the radiative boundary and estimate the solution at these locations using the Walk-on-Stars method with 256 random walks per point. The relaxation coefficient is fixed to $0.25$ for all runs. Boundary geometries used in the evaluation are shown in Figure~\ref{fig:geometry}.

In our notation, iteration~1 corresponds to a single linearization of the radiative boundary condition using the initial proxy solution, without fixed-point refinement. Mean squared error (MSE) values on the radiative boundary after each iteration are reported in Table~\ref{tab:iteration_error}. The error is computed over the sampled boundary points with respect to the analytical ground-truth solution.

\begin{figure*}
    \centering
    \subfigbottomskip=2pt
    \subfigcapskip=-5pt
        \subfigure[]{
        \includegraphics[width=0.24\linewidth]{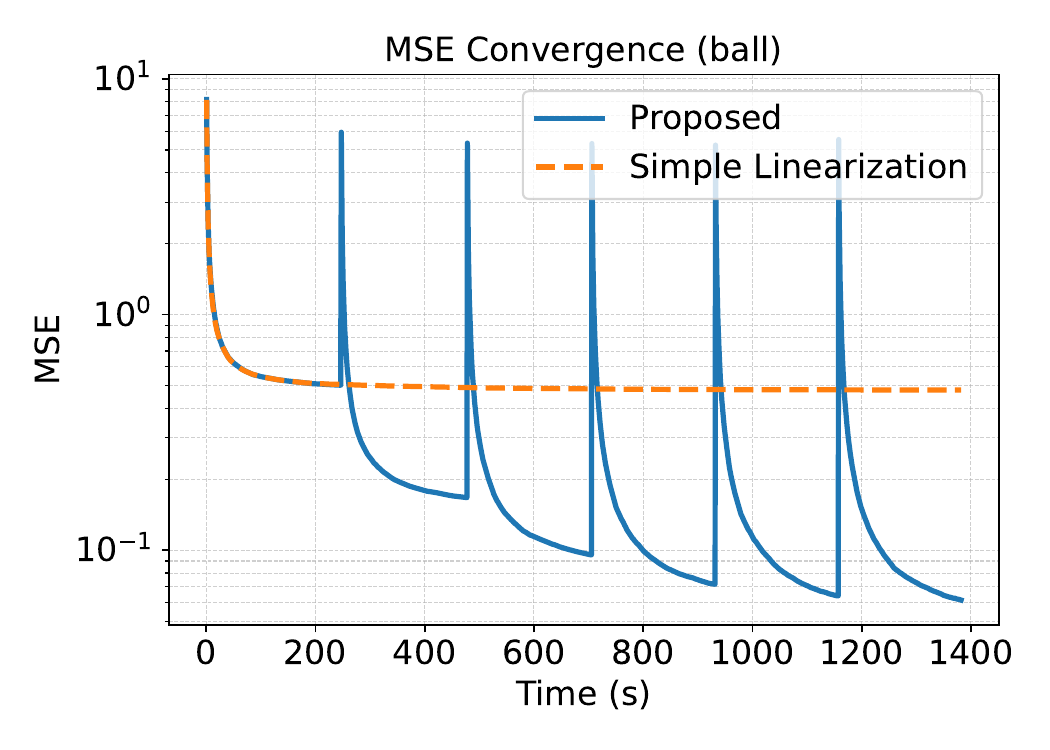}}
    \subfigure[]{
        \includegraphics[width=0.24\linewidth]{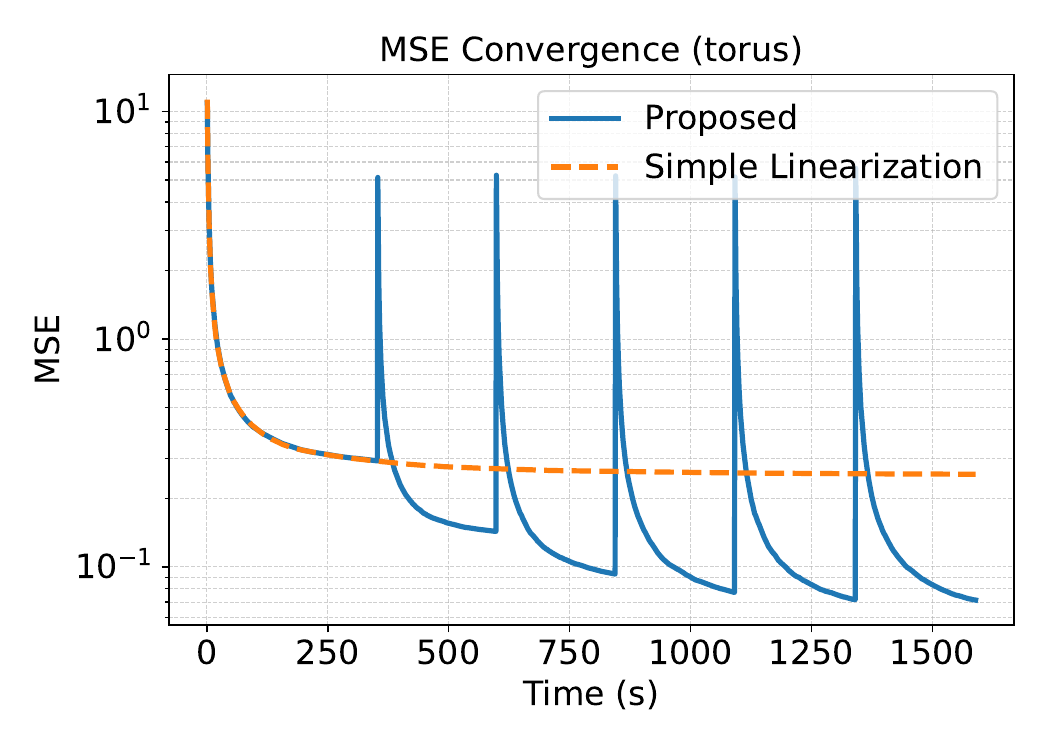}}
    \subfigure[]{
        \includegraphics[width=0.24\linewidth]{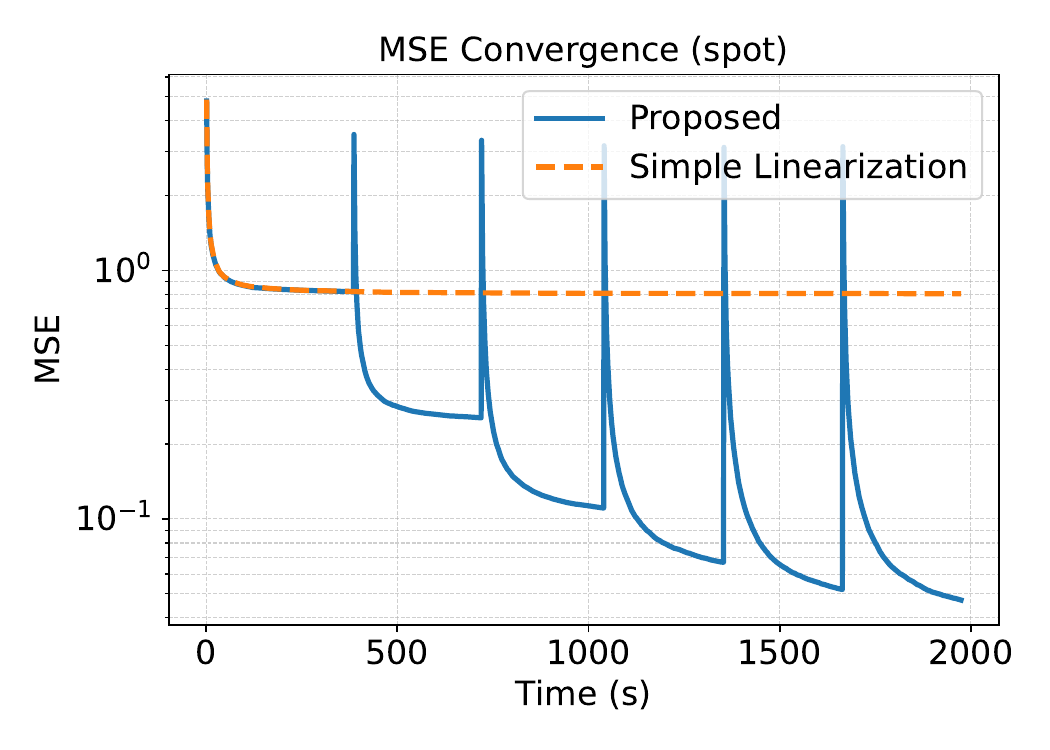}}
    \subfigure[]{
        \includegraphics[width=0.24\linewidth]{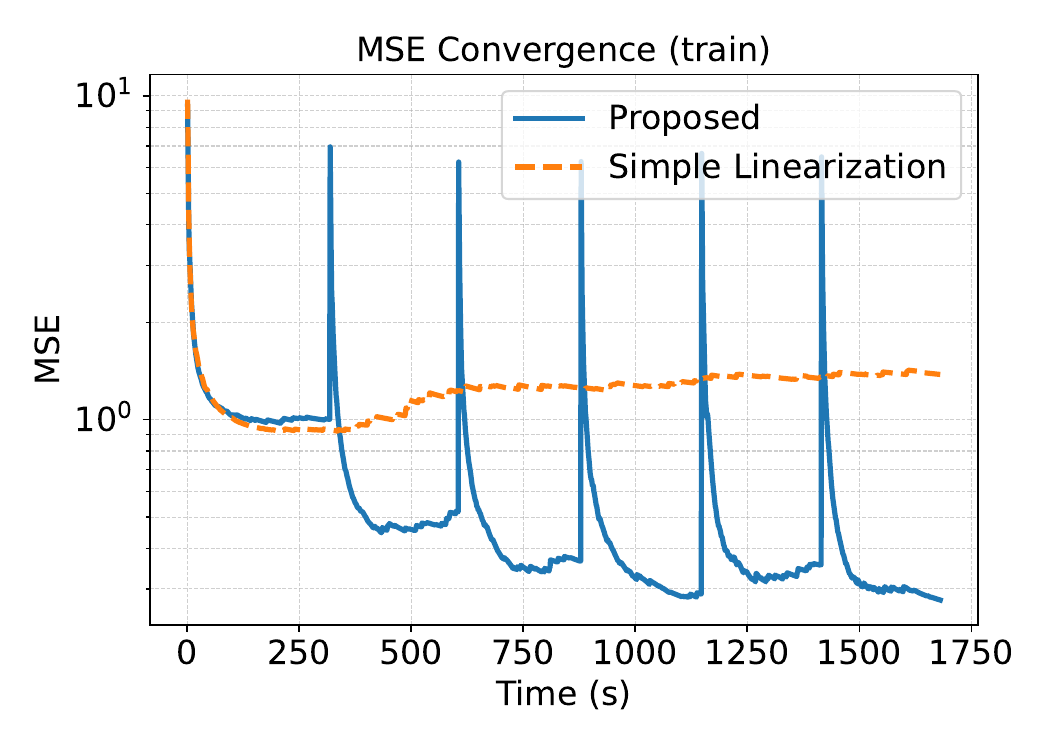}}
    \caption{Mean squared error (MSE) versus computation time for the proposed fixed-point Monte Carlo solver on four representative geometries: (a) Ball, (b) Torus, (c) Spot, and (d) Train. Each curve shows the convergence behavior across successive Picard iterations under a fixed computational budget per iteration. The results demonstrate rapid error reduction in early iterations and stable convergence over time, highlighting the efficiency and robustness of the proposed approach across different geometric complexities.}
    \label{fig:mse_curve}
\end{figure*}

Across all tested geometries, the proposed Picard-style fixed-point iteration consistently reduces the mean squared error. Empirically, convergence is observed within approximately 5--10 iterations. Compared to a single linearization step (iteration~1), the fixed-point scheme achieves substantially lower error even when initialized with an inaccurate proxy solution.

Figure~\ref{fig:denoise} provides a qualitative visualization of the boundary solution estimates before and after denoising. In addition to the fixed-point iteration, we apply the proposed heteroscedastic regression-based denoising method to the boundary estimates. The resulting denoised errors, reported in the final row of Table~\ref{tab:iteration_error}, exhibit \changed{a} variance reduction of approximately one to two orders of magnitude across all geometries. This demonstrates that boundary denoising effectively suppresses Monte Carlo variance and further stabilizes the nonlinear iteration.

\begin{table}[h]
\centering
\caption{Iteration-wise convergence behavior of the proposed fixed-point Monte Carlo solver on four representative geometries. Each entry reports the boundary estimation error on the radiative boundary $\partial\Omega_{\changedd{R}}$ at a given iteration. The row labeled \emph{Denoised} shows the error after applying the heteroscedastic regression-based boundary denoising described in Section~\ref{sec:denoising}. The results demonstrate rapid error reduction across iterations and a significant additional improvement from the proposed denoising strategy.}
\label{tab:iteration_error}
\begin{tabular}{ccccc}
\toprule
Iteration & Ball & Torus & Spot & Train \\
\midrule
1 & 0.5014 & 0.2917 & 0.8287 & 1.0014 \\
2 & 0.1671 & 0.1430 & 0.2550 & 0.5191 \\
3 & 0.0954 & 0.0929 & 0.1105 & 0.3652 \\
4 & 0.0716 & 0.0771 & 0.0669 & 0.2891 \\
5 & 0.0640 & 0.0717 & 0.0520 & 0.3558 \\
6 & 0.0614 & 0.0714 & 0.0471 & 0.2770 \\
\midrule
Denoised & 0.0031 & 0.0032 & 0.0021 & 0.0063 \\
\midrule
MC Runtime (s)  & 1381.9 & 1589.4 & 1971.1 & 1680.5 \\
Denoise time (s) & 16.7 & 16.81 & 20.1 & 18.4 \\
\bottomrule
\end{tabular}
\end{table}

\subsection{Light-Heat Coupled Problems: Radiation in Vacuum}
\label{sec:coupled}

We next demonstrate the applicability of the proposed Monte Carlo framework to a coupled light-heat problem arising in vacuum environments. In this setting, radiative transfer governs energy exchange at the surface, while heat conduction redistributes absorbed energy inside the solid. The two processes are strongly coupled through a nonlinear radiative boundary condition, making the problem challenging for conventional mesh-based solvers.

\paragraph{Physical Setup.}
We consider an irregular geometry placed in vacuum and illuminated by a distant directional light source that mimics solar radiation. Since the surrounding medium is vacuum, no convective heat transfer is present, and energy exchange occurs exclusively through radiation at the surface and conduction within the asteroid. The interior temperature field $T$ satisfies the steady-state heat equation
\begin{equation}
    -\nabla \cdot (k \nabla T) = 0,
\end{equation}
where $k$ denotes the thermal conductivity. On the surface, the radiative boundary condition is given by
\begin{equation}
    -k \frac{\partial T}{\partial \vec{n}}(\mathbf{x})
    = \epsilon \sigma T^4(\mathbf{x}) - \Phi_{\mathrm{in}}(\mathbf{x}),
\end{equation}
where $\epsilon$ is the surface emissivity and $\sigma$ is the Stefan--Boltzmann constant. 
The incident radiative flux $\Phi_{\mathrm{in}}(\mathbf{x})$ is induced by a parallel light source and is given by
\begin{equation}
    \Phi_{\mathrm{in}}(\mathbf{x})
    = L_0 \, (\boldsymbol{\omega}_0 \cdot \vec{n}(\mathbf{x}))_+ \, V(\mathbf{x}, \boldsymbol{\omega}_0),
\end{equation}
where $L_0$ is the constant incident radiative intensity, $\boldsymbol{\omega}_0$ denotes the fixed incoming light direction, $\vec{n}(\mathbf{x})$ is the outward surface normal, and $V(\mathbf{x}, \boldsymbol{\omega}_0) \in \{0,1\}$ is a visibility indicator computed by ray tracing, which accounts for self-shadowing. The operator $(\cdot)_+ = \max(\cdot, 0)$ enforces the cosine law. For the initial guess, we use 
\begin{equation}
    T_0 = \left( \frac{L_0}{4 \epsilon \sigma} \right)^{1/4}.
\end{equation}

\paragraph{Monte Carlo Coupling Strategy.}
Radiative transfer is simulated using Monte Carlo ray tracing to estimate the incident flux $\Phi_{\text{in}}$ at surface locations. Heat conduction inside the asteroid is solved using the walk-on-stars method, which avoids volumetric meshing and operates directly on the continuous geometry. The nonlinear coupling induced by the $T^4$ emission term is resolved through a Picard-style fixed-point iteration. At each iteration, the radiative flux and the conductive solution are alternately updated until convergence. This coupling is performed entirely at the boundary, preserving the processing-free nature of the solver and avoiding explicit discretization of the interior domain.

\begin{figure*}
    \centering
    \includegraphics[width=\linewidth]{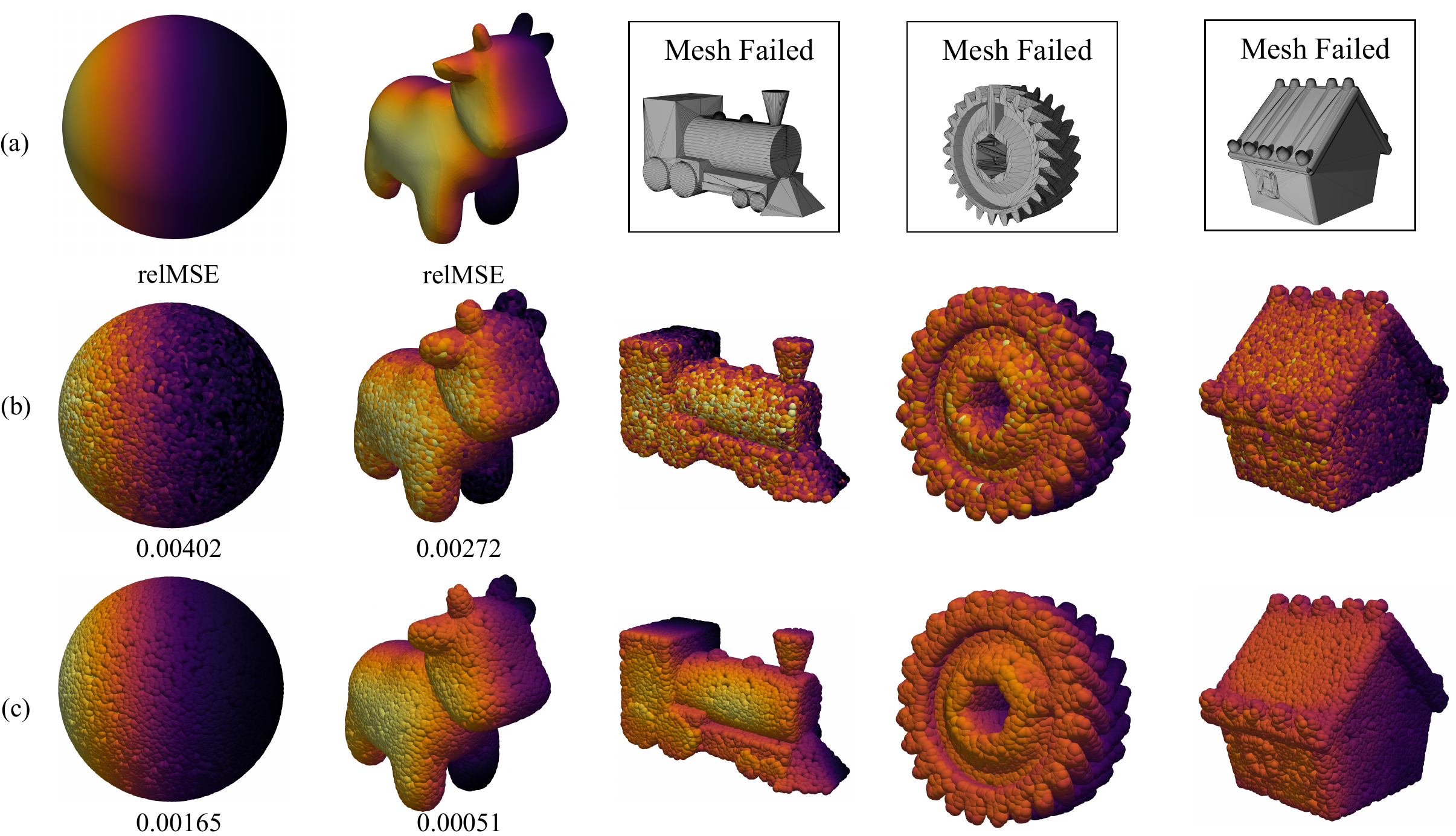}
    \caption{Comparison on a coupled light-heat problem. \changed{Here we are visualizing the MLS point cloud.} A parallel light source illuminates a solid sphere modeled as an ideal black body ($\epsilon = 1$, $k = 1$, $L_0 = 1000 \mathrm{W/m^2}$). (a) Reference FEM solution with radiative boundary conditions provided by COMSOL Multiphysics\textregistered{} \cite{comsol2025multiphysics}. (b) Result of the proposed Monte Carlo solver before denoising. (c) Result after heteroscedastic regression denoising. The proposed method closely matches the FEM solution on \changedd{simplified} geometries like the sphere and the spot model. For more complex geometries, COMSOL Multiphysics\textregistered{} \cite{comsol2025multiphysics} cannot generate \changed{high-quality volumetric meshes} automatically, while the proposed method remains robust and produces plausible results.}
    \label{fig:blackball}
\end{figure*}

\paragraph{Validation.}
To validate the correctness of the proposed coupled formulation, we compare our method against a classical finite element method (FEM) solver with radiative boundary conditions. 
All test geometries are illuminated by a parallel light source and modeled as ideal black bodies. 
As shown in Figure~\ref{fig:blackball}, the temperature fields produced by the proposed method closely match the FEM reference solutions generated by commercial software COMSOL Multiphysics\textregistered{} \cite{comsol2025multiphysics} for moderately complex geometries such as the sphere and the spot model.

For geometries with highly irregular topology and fine-scale features, such as the train model, generating a high-quality volumetric mesh suitable for FEM becomes challenging, which limits the applicability of FEM solvers in these cases. 
In contrast, the proposed Monte Carlo approach operates directly on surface geometry and does not require volumetric meshing, enabling robust simulation on complex models.

For each example, we use 128 Monte Carlo samples per iteration and perform six fixed-point iterations, resulting in a total runtime of approximately 30 minutes on our hardware. 
In comparison, FEM produces solutions within approximately one minute for geometries where meshing is feasible, but remains sensitive to boundary mesh quality and discretization.
These results highlight the complementary nature of the two approaches: FEM offers high efficiency on well-conditioned meshes, while our method provides superior robustness and geometric flexibility for complex models.

\paragraph{Asteroid Example.}
Compared to FEM solvers, Monte Carlo PDE methods offer significantly greater flexibility for handling complex geometries without volumetric meshing. To demonstrate this advantage, Figure~\ref{fig:coupled} visualizes the steady-state surface temperature field of an asteroid model consisting of approximately 440k triangles. The asteroid is modeled as a gray body with emissivity $\epsilon = 0.9$ and thermal conductivity $k_0 = 1$, and is illuminated by a parallel light source with irradiance $L_0 = 1073~\mathrm{W/m^2}$, corresponding to solar radiation in orbit.

The visualization in Figure~\ref{fig:coupled} is generated using a deferred shading pipeline~\cite{deering1988triangle} for surface sampling. World-space positions are obtained from the rasterizer, and temperatures are evaluated at these locations using the MLS boundary representation described in Section~\ref{sec:boundary_representation}. Regions directly exposed to the light source reach significantly higher temperatures, while shadowed regions remain substantially cooler. Despite the stochastic nature of the Monte Carlo estimator and the strong nonlinearity introduced by radiative emission, the resulting temperature field varies smoothly across the surface.

To validate the solution, we compare against a FEM reference generated using COMSOL Multiphysics\textregistered{} \cite{comsol2025multiphysics}. Due to the geometric complexity of the asteroid, direct tetrahedral meshing of the original model is not feasible; instead, a simplified geometry is used for FEM simulation, resulting in a system with approximately 405k degrees of freedom. As shown in Figure~\ref{fig:coupled}, the proposed method produces a temperature field consistent with the FEM solution on the original geometry, while the vanilla linearization approach exhibits significant deviation.

Empirically, the fixed-point iteration converges within 5--10 iterations, even under the strong nonlinearity of the radiative boundary condition. This experiment highlights the ability of the proposed framework to solve coupled light--heat problems on complex geometries in vacuum without volumetric meshing, view-factor precomputation, or domain decomposition, making it particularly suitable for large-scale and geometrically intricate scenarios where traditional deterministic solvers become cumbersome.

\begin{figure*}
    \centering
    \subfigbottomskip=2pt
    \subfigcapskip=-5pt
        \subfigure[]{
        \includegraphics[width=0.24\linewidth]{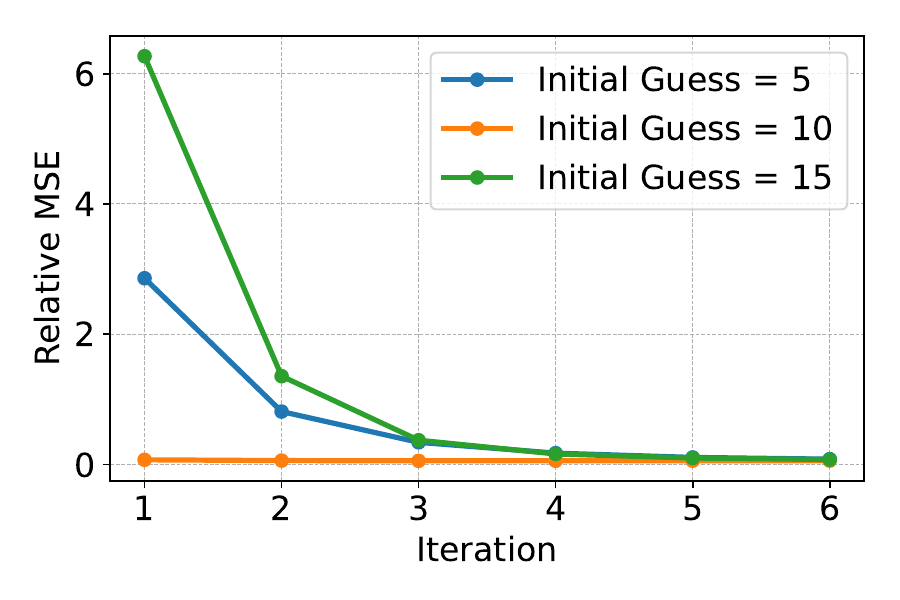}}
    \subfigure[]{
        \includegraphics[width=0.24\linewidth]{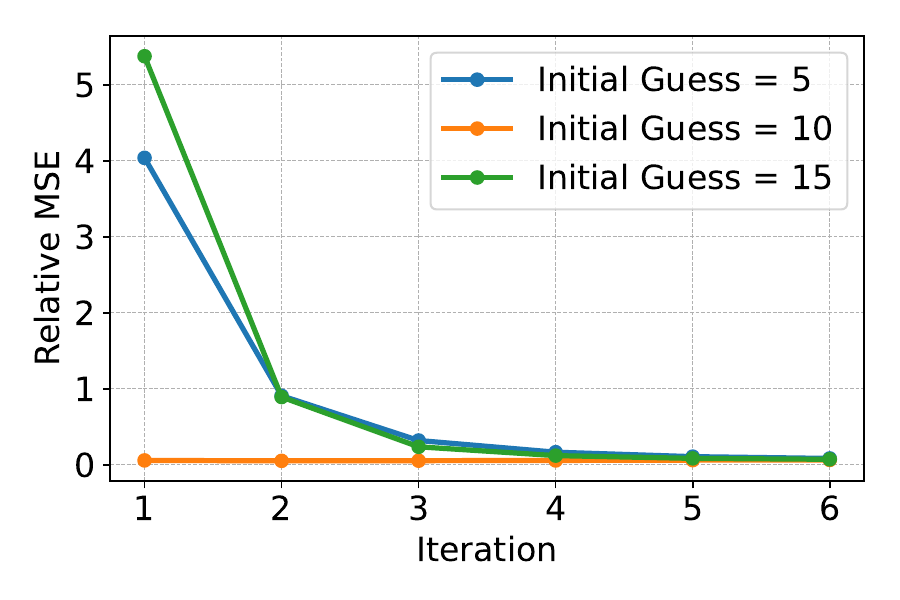}}
    \subfigure[]{
        \includegraphics[width=0.24\linewidth]{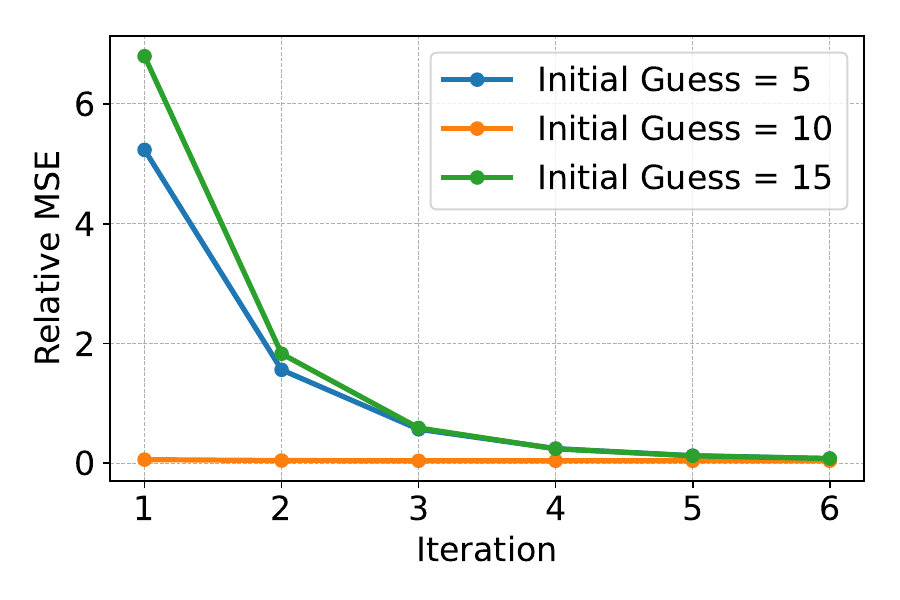}}
    \subfigure[]{
        \includegraphics[width=0.24\linewidth]{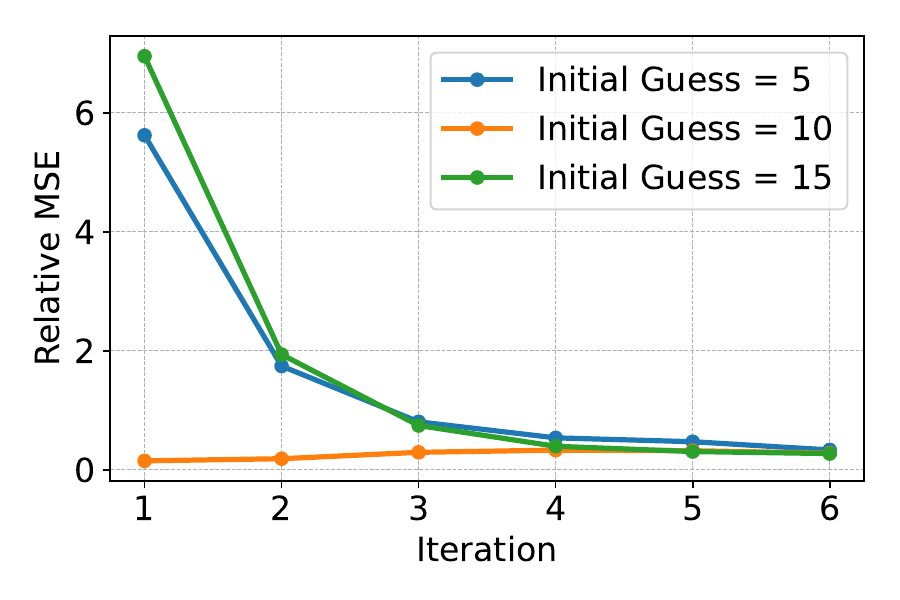}}
    \caption{Effect of the initial proxy solution on convergence accuracy. 
    Mean squared error (MSE) versus computation time is shown for four geometries: (a) Ball, (b) Torus, (c) Spot, and (d) Train. When the initial proxy solution is close to the ground truth, a single-step linearization yields competitive accuracy. However, as the initial proxy deviates from the true solution, linearization exhibits persistent bias and elevated error. In contrast, the proposed fixed-point iteration consistently reduces error over time and converges to significantly more accurate solutions across all geometries.}

    \label{fig:initial_ablation}
\end{figure*}

\begin{figure*}
    \centering
    \subfigbottomskip=2pt
    \subfigcapskip=-5pt
        \subfigure[]{
        \includegraphics[width=0.24\linewidth]{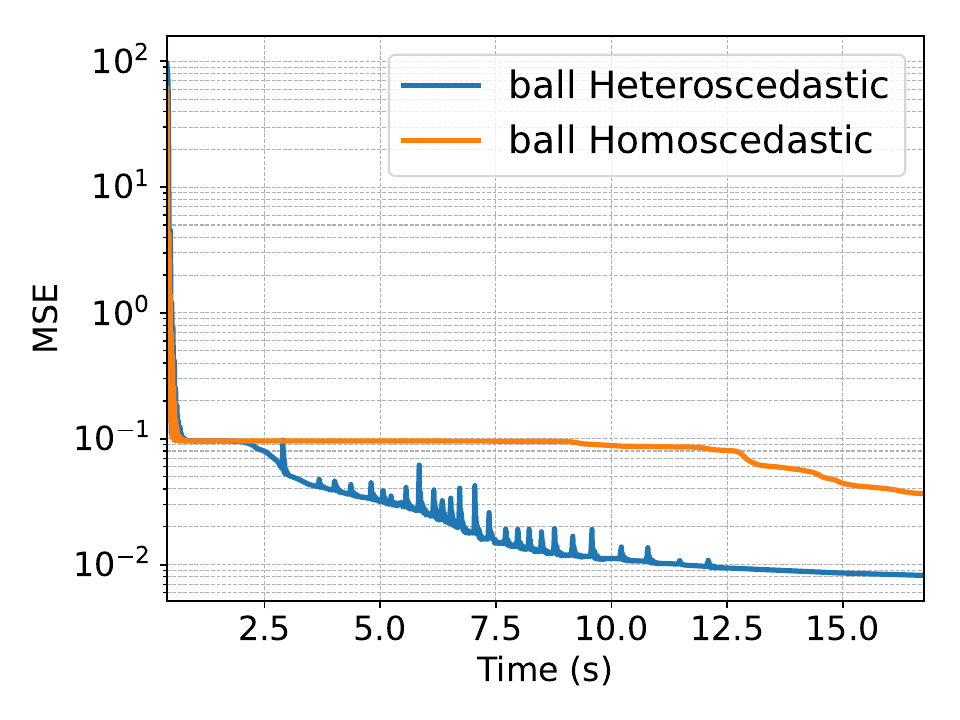}}
    \subfigure[]{
        \includegraphics[width=0.24\linewidth]{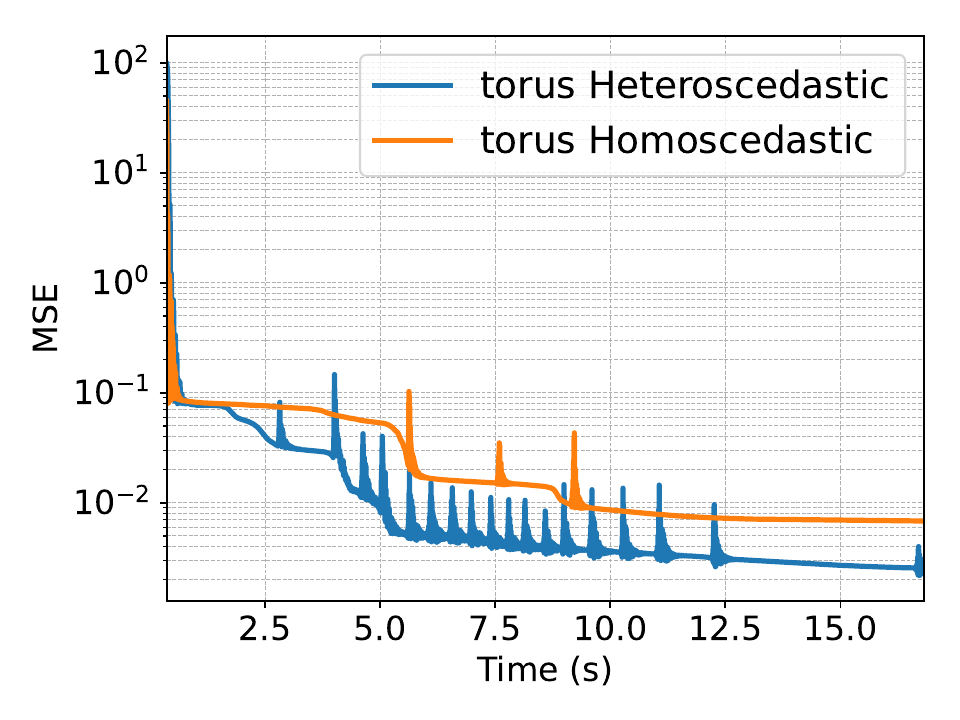}}
    \subfigure[]{
        \includegraphics[width=0.24\linewidth]{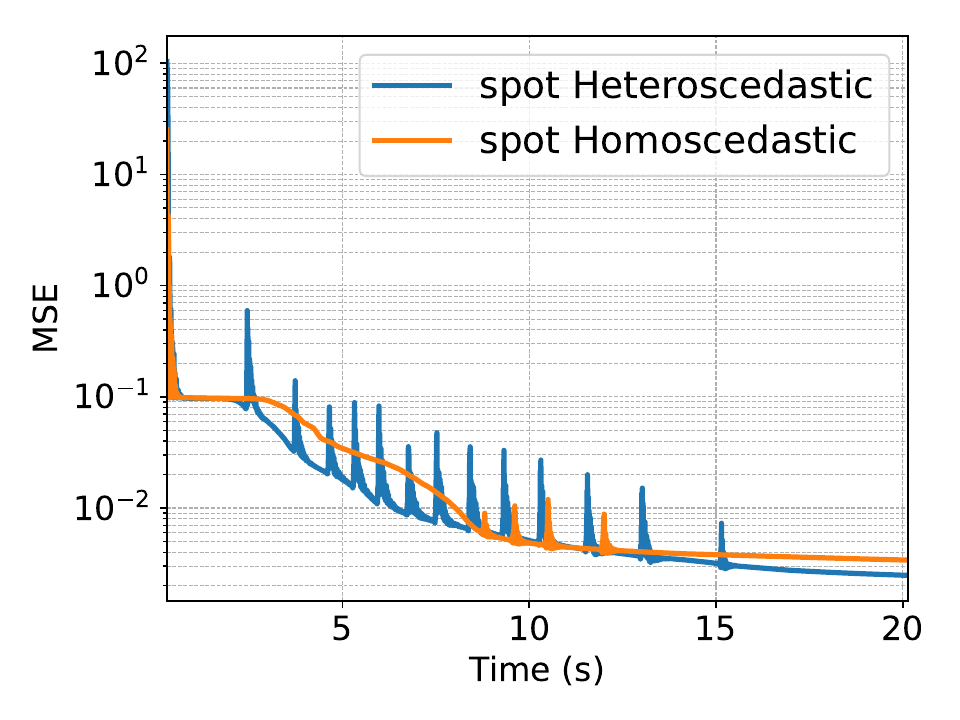}}
    \subfigure[]{
        \includegraphics[width=0.24\linewidth]{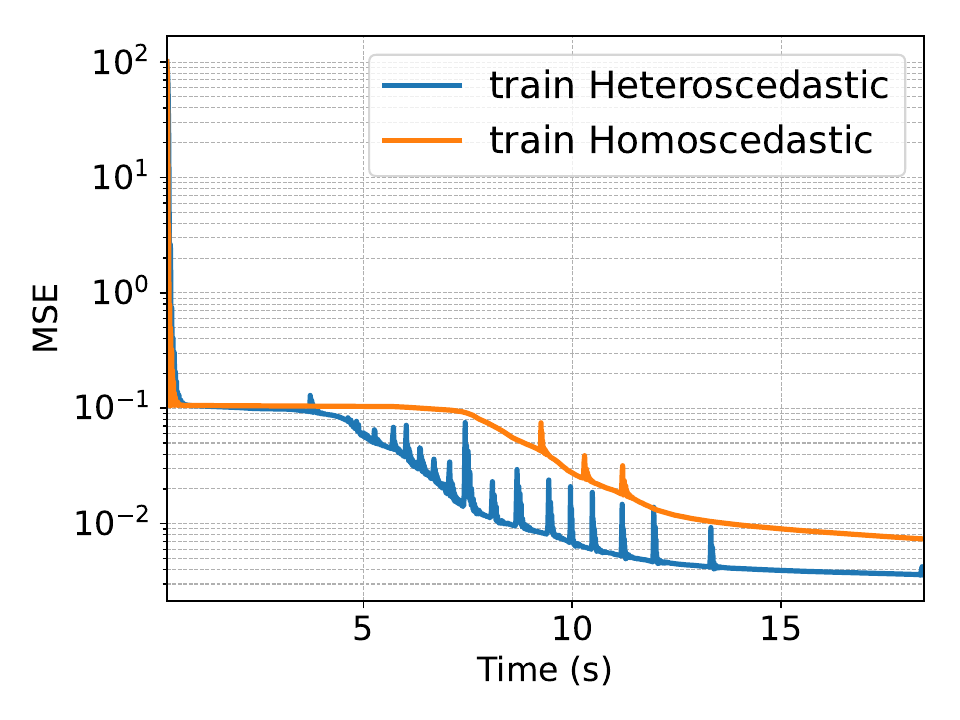}}
    \caption{Comparison of heteroscedastic and homoscedastic regression for boundary denoising. Mean squared error (MSE) versus computation time is shown for four geometries: (a) Ball, (b) Torus, (c) Spot, and (d) Train. The heteroscedastic regression consistently converges faster and reaches a lower error, highlighting the benefit of explicitly modeling spatially varying Monte Carlo noise.}
    \label{fig:regression_ablation}
\end{figure*}

\subsection{Ablation and Effect of Hyperparameters}
\label{sec:ablation}

We analyze the influence of key hyperparameters in the proposed fixed-point Monte Carlo framework, with a particular focus on the initialization of the radiative boundary temperature and the relaxation coefficient used in the Picard iteration. These parameters directly affect the stability and convergence behavior of the nonlinear boundary coupling.

\paragraph{Effect of Initial Proxy.}
We evaluate the robustness of the proposed method with respect to the initial proxy and compare it against the simple linearization strategy. 
Following the setup in Section~\ref{sec:synthetic}, we test three initial values (5, 10, and 15), where 10 corresponds to the \changed{best-aligned} proxy by construction. 
As shown in Figure~\ref{fig:initial_ablation}, when the proxy solution is close to the true solution, the linearized solve (one iteration) produces an accurate result. 
However, as the proxy deviates from the ground truth, the error after the first iteration increases substantially and remains biased. 
In contrast, the proposed fixed-point iteration progressively corrects this bias and consistently converges to lower-error solutions, demonstrating strong robustness to inaccurate initialization.

\paragraph{Effect of the Relaxation Coefficient.}
To evaluate the impact of relaxation on convergence, we repeat the black-body radiation experiment described in Section~\ref{sec:coupled} using different relaxation coefficients~$\alpha$. Figure~\ref{fig:relaxation_ablation_mse} reports the relative mean squared error (relMSE) at each iteration, while Figure~\ref{fig:relaxation_ablation} visualizes the intermediate solutions over iterations.

As shown in the results, smaller relaxation coefficients lead to more stable convergence. In contrast, larger values (e.g., $\alpha = 0.5$) may induce oscillatory behavior or stagnation due to the strong nonlinearity of the radiative boundary condition. These observations are consistent with the theoretical instability of unrelaxed Picard iterations in nonlinear boundary value problems.

\begin{figure}
    \centering
    \includegraphics[width=\linewidth]{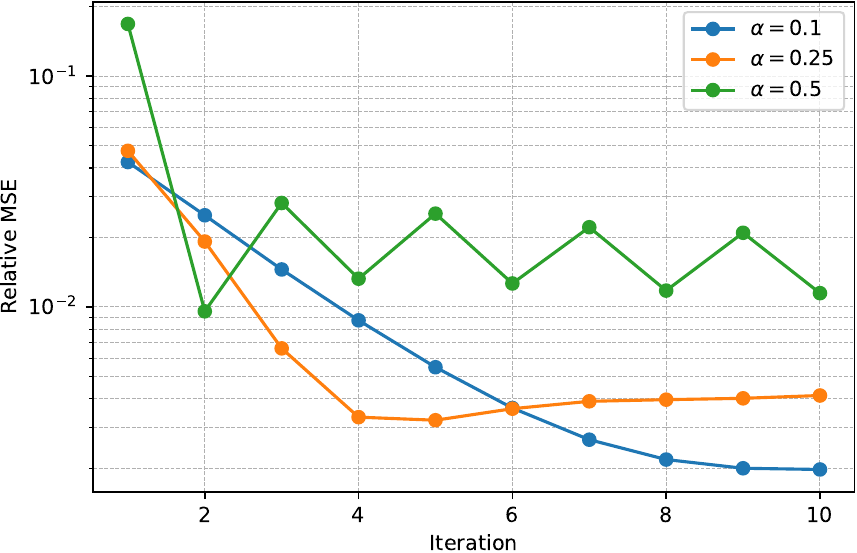}
    \caption{Iteration-wise relMSE for the black-body radiation experiment under different relaxation coefficients~$\alpha$. A small relaxation coefficient ($\alpha = 0.1$) converges slowly but achieves the lowest final error. A moderate value ($\alpha = 0.25$) \changed{converges fastest} but plateaus at a slightly higher error. A large relaxation coefficient ($\alpha = 0.5$) converges the slowest and exhibits oscillatory behavior, preventing stable error reduction.}
    \label{fig:relaxation_ablation_mse}
\end{figure}

\begin{figure*}
    \includegraphics[width=\linewidth]{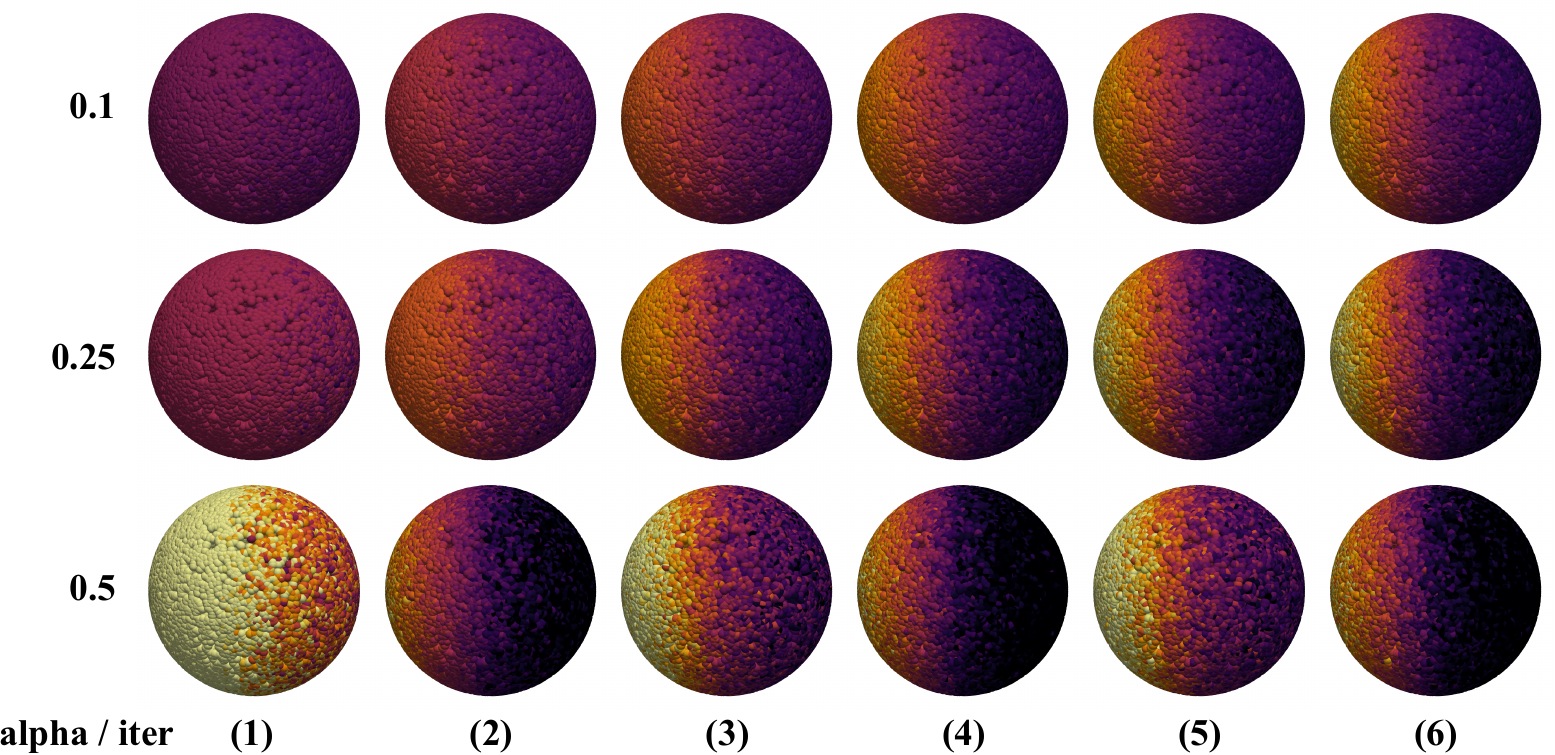}
    \caption{Intermediate solutions at iterations 0--6 for different relaxation coefficients~$\alpha$. \changed{Here we are visualizing the MLS point cloud.} 
For $\alpha = 0.1$ and $\alpha = 0.25$, the solution increases monotonically from an underestimated initial state and converges steadily toward the reference solution. 
In contrast, $\alpha = 0.5$ exhibits pronounced underestimate--overestimate oscillations, indicating instability of the unrelaxed Picard iteration under strong radiative nonlinearity.}
    \label{fig:relaxation_ablation}
\end{figure*}

\paragraph{Heteroscedastic vs. Homoscedastic Regression.}
We compare the proposed heteroscedastic regression-based denoiser with a homoscedastic baseline that assumes spatially constant noise variance. Both methods use the same network architecture, training schedule, and Monte Carlo sampling data. Figure~\ref{fig:regression_ablation} shows the mean squared error (MSE) as a function of computation time for both approaches. Across all tested geometries, the heteroscedastic model converges faster and achieves a lower squared error, demonstrating its effectiveness in capturing the spatially varying noise characteristics inherent in Monte Carlo boundary estimators.

\paragraph{Discussion.}
Together, these ablation studies confirm that the proposed fixed-point Monte Carlo solver is robust with respect to both initialization and relaxation parameters. While the nonlinearity of radiative boundary conditions inevitably introduces sensitivity in early iterations, the combination of Picard iteration and relaxation effectively mitigates these issues, enabling reliable convergence without adaptive meshing or intrusive preprocessing. This property is particularly important for preserving the \emph{processing-free} nature of Monte Carlo PDE solvers.

\subsection{Generalization and Failure Cases}
\label{sec:failure_cases}

\changed{Our fixed-point iteration framework is not limited to radiative boundary conditions and can be applied to a broader class of nonlinear boundary value problems. As discussed in Section~\ref{sec:iteration}, the proposed iteration scheme can be applied to other boundary conditions. In this part, we present results on nonlinear Robin boundary conditions with varying exponents in the radiative term. Specifically, we present results with boundary conditions of the form
\begin{equation}
    \frac{\partial u}{\partial \vec{n}}(x) + \gamma (u^4(x) + u^2\sqrt{u}) = h(x)
\end{equation}
The results are shown in Figure~\ref{fig:other_nonlinear}, demonstrating the generality and robustness for a wide range of nonlinear boundary conditions.}

\begin{figure}
    \includegraphics[width=\linewidth]{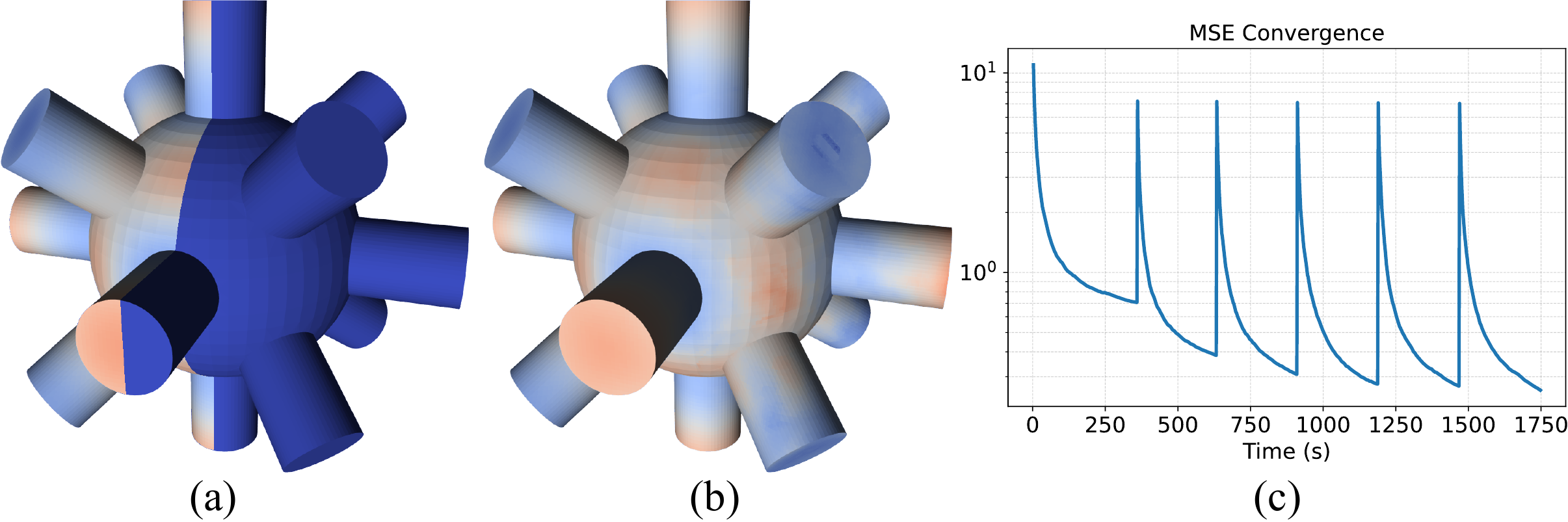}
    \caption{Results for other nonlinear boundary conditions. (a) Following the setup in Section~\ref{sec:synthetic}, the left part is a Dirichlet boundary and the right part is a nonlinear Robin boundary with a combination of $u^4$ and $u^{2.5}$ terms. (b) Result of the proposed method. (c) MSE curve across iterations.}
    \label{fig:other_nonlinear}
\end{figure}

\changed{
The proposed fixed-point iteration is not guaranteed to converge for all nonlinear boundary conditions or under all parameter settings. In particular, when the relaxation coefficient $\alpha$ is too large, the iteration may exhibit oscillatory behavior and fail to converge to a stable solution. Figure~\ref{fig:failure} illustrates such a failure case, where the error fluctuates across iterations without consistent reduction, indicating divergence of the iteration. Careful selection of the relaxation parameter and monitoring of convergence are essential to ensure stability in practice.
}

\begin{figure*}
    \includegraphics[width=\linewidth]{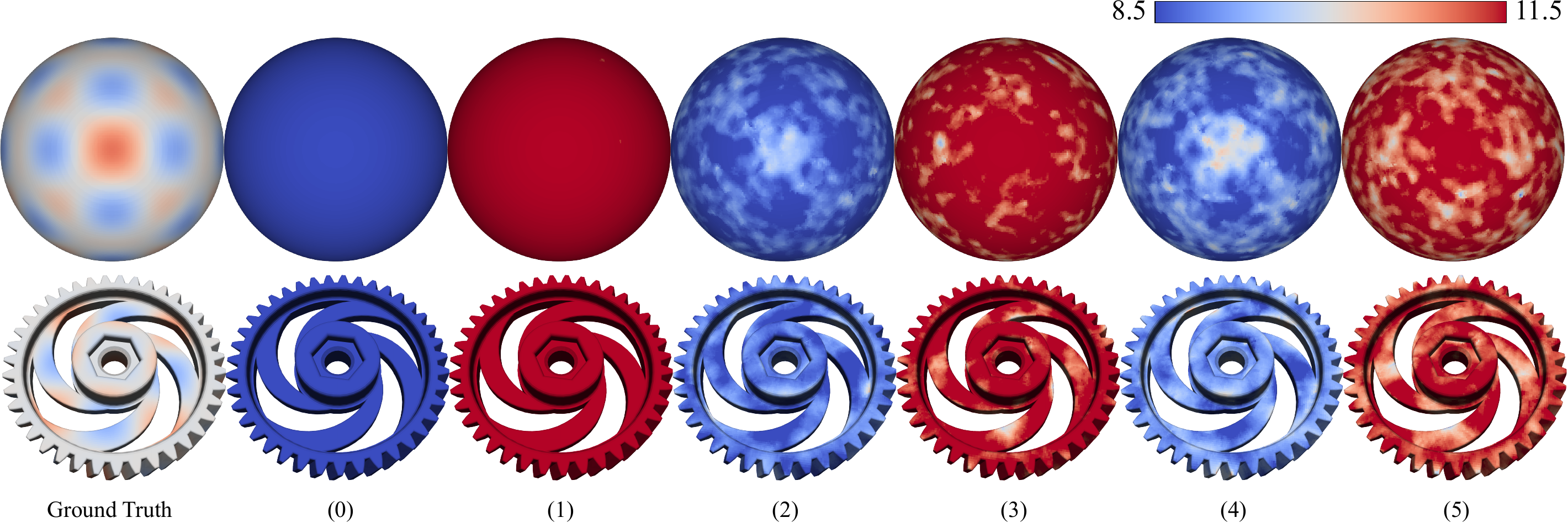}
    \caption{Failure cases of the fixed-point iteration. When the relaxation coefficient is too large (e.g., $\alpha = 0.5$), the iteration exhibits oscillatory behavior and fails to converge to a stable solution. }
    \label{fig:failure}
\end{figure*}

\section{Conclusion and Future Work}

In this work, we presented a fixed-point iteration framework for solving boundary value problems with nonlinear radiative boundary conditions using Monte Carlo PDE solvers. We conclude by discussing current limitations and outlining promising directions for future research.

\paragraph{Acceleration.}
While we introduced a heteroscedastic regression method to denoise on-boundary solution estimates, the overall framework can benefit from additional variance reduction and acceleration techniques. Although random walks are initiated on the boundary, they frequently traverse the interior of the domain, where existing methods such as harmonic caching~\cite{zhou2025harmonic}, path guiding~\cite{huang2025guiding}, and neural caching~\cite{li2023neural} could be incorporated. Integrating these techniques into the proposed boundary-focused iteration framework remains an interesting direction for future work. Anderson acceleration~\cite{anderson1965iterative} and its variants~\cite{fang2009two, walker2011anderson} can also be incorporated into the fixed-point iteration.

\paragraph{Bias.} More principled debiasing strategies are possible. For example, independent sampling of higher-order moments (e.g., $T^3$) or techniques like \changed{\citet{misso2022unbiased} could} be incorporated to eliminate this bias. We leave such extensions to future work.

\paragraph{Transient Problems.}
This paper focuses on steady-state heat transfer and fills an important gap by supporting nonlinear radiative boundary conditions. Extending the proposed approach to transient heat transfer problems is nontrivial and remains unexplored. In particular, WoSt-type algorithms~\cite{sawhney2023walk, miller2024walkin} are fundamentally designed for elliptic problems and do not naturally generalize to time-dependent settings. Developing Monte Carlo solvers that can jointly handle temporal evolution and nonlinear boundary interactions is an open challenge.

\paragraph{Physical Model.}
This work focuses on incorporating radiative boundary conditions into Monte Carlo PDE solvers under the assumption of constant, gray-body emissivity. Several physically more sophisticated radiation models remain unexplored. In practice, the emissivity $\epsilon$ may depend on temperature, material properties, or wavelength, introducing additional nonlinearity and coupling across spectral bands. Modeling spectrally varying or temperature-dependent emissivity would require extending the current formulation beyond the scalar Stefan-Boltzmann law. We leave the integration of these more complex physical models to future work.

In summary, we proposed a fixed-point iteration-based Monte Carlo PDE solver for radiative heat transfer that preserves the geometric flexibility of Monte Carlo methods and applies directly to complex geometries without preprocessing. Compared to simple linearization strategies, the proposed approach empirically reaches significantly improved accuracy from imprecise initial guesses, but it does not provide a global theoretical convergence guarantee (Section~\ref{sec:failure_cases}, Figure~\ref{fig:failure}). We believe this work opens the door to a broader class of nonlinear boundary problems in Monte Carlo-based PDE solvers.

\bibliographystyle{ACM-Reference-Format}
\bibliography{ref}

\appendix

\section{Proof of Theorems}

\subsection{Proof of Lemma \ref{thm:comparison}}
\label{sec:proof1}
\begin{proof}
Define $r(x) = u_0(x) - u(x)$. Subtracting the governing equations yields
\begin{equation}
\Delta r = 0 \quad \text{in } \Omega,
\end{equation}
and
\begin{equation}
r = 0 \quad \text{on } \partial\Omega_D.
\end{equation}

On the radiative boundary $\partial\Omega_R$, we have
\begin{equation}
\frac{\partial r}{\partial n}
= -\mu p^3 u_0 + \mu u^4.
\end{equation}
Rearranging terms gives
\begin{equation}
\frac{\partial r}{\partial n} + \mu p^3 r
= \mu \bigl(u^4 - p^3 u\bigr)
= \mu u \bigl(u^3 - p^3\bigr).
\end{equation}
Since $p(x) \le u(x)$ on $\partial\Omega_R$, the right-hand side is nonnegative, and therefore
\begin{equation}
    \label{proof:robin}
    \frac{\partial r}{\partial n} + \mu p^3 r \ge 0
    \quad \text{on } \partial\Omega_R.
\end{equation}

Since $\Delta r = 0$, by the maximum principle of harmonic functions, the minimum value of $r$ in $\bar{\Omega}$ appears on the boundary $\partial \Omega$. Suppose the position taking the minimum value is $x_0$, then
\begin{equation}
    \frac{\partial r}{\partial \vec{n}}(x_0) < 0.
\end{equation}
If $r(x_0) < 0$, we have
\begin{equation}
    \frac{\partial r}{\partial \vec{n}}(x_0) + \mu p^3 r(x_0) < 0,
\end{equation}
since $\mu \geq 0$ and $p \geq 0$, which \changed{contradicts} the inequality \eqref{proof:robin}. Thus, we claim that $r(x) \geq 0$ for all $x \in \bar{\Omega}$.
\end{proof}

\section{Moving Least Squares Surface Representation}
\label{sec:appendix_mls}

To evaluate proxy solution values at arbitrary surface query points, we approximate the boundary proxy solution field using a moving least squares (MLS) reconstruction over Monte Carlo samples.

Given a set of boundary samples $\{(\mathbf{x}_i, T_i)\}_{i=1}^N$, the temperature at a query point $\mathbf{x}$ is approximated by a local linear model
\begin{equation}
T(\mathbf{x} + \mathbf{q}) \approx \mathbf{p}(\mathbf{q})^\top \boldsymbol{\theta}, \quad 
\mathbf{p}(\mathbf{q}) = [1, q_1, \dots, q_d]^\top,
\end{equation}
where $\mathbf{q}_i = \mathbf{x}_i - \mathbf{x}$. The coefficients $\boldsymbol{\theta}$ are obtained by solving the weighted least squares problem
\begin{equation}
\min_{\boldsymbol{\theta}} \sum_{i \in \mathcal{N}(\mathbf{x})} 
w(\|\mathbf{q}_i\|)\left(\mathbf{p}(\mathbf{q}_i)^\top \boldsymbol{\theta} - T_i\right)^2,
\end{equation}
with Gaussian kernel
\begin{equation}
w(r) = \exp\left(-\frac{r^2}{h^2}\right).
\end{equation}

The neighborhood $\mathcal{N}(\mathbf{x})$ consists of all samples within radius $r$ of $\mathbf{x}$. If fewer than four neighbors are found, we fall back to local averaging for numerical robustness. The temperature estimate is given by the constant term $T(\mathbf{x}) = \max(\theta_0, 0)$.

Algorithm~\ref{alg:mls} summarizes the MLS evaluation procedure.

\begin{algorithm}[]
    \caption{Moving Least Squares (MLS) Evaluation}
    \label{alg:mls}
    \begin{algorithmic}[1]
    \Require Query point $\mathbf{x} \in \mathbb{R}^d$, point cloud $\{\mathbf{x}_i, T_i\}_{i=1}^N$, neighborhood radius $r$, kernel width $h$
    \Ensure Estimated value $T(\mathbf{x})$
    
    \State Initialize empty neighbor set $\mathcal{N} \gets \emptyset$
    
    \For{$i = 1$ to $N$}
        \If{$\|\mathbf{x}_i - \mathbf{x}\| \le r$}
            \State Add index $i$ to $\mathcal{N}$
        \EndIf
    \EndFor
    
    \If{$|\mathcal{N}| < 4$}
        \State \Return $\frac{1}{|\mathcal{N}|} \sum_{i \in \mathcal{N}} T_i$
    \EndIf
    
    \State Define linear basis $\mathbf{p}(\mathbf{q}) = [1, q_1, q_2, \dots, q_d]^\top$
    
    \For{each $i \in \mathcal{N}$}
        \State $\mathbf{q}_i \gets \mathbf{x}_i - \mathbf{x}$
        \State $w_i \gets \exp\left(-\frac{\|\mathbf{q}_i\|^2}{h^2}\right)$
        \State $\mathbf{A}_i \gets \mathbf{p}(\mathbf{q}_i)$
        \State $f_i \gets T_i$
    \EndFor
    
    \State Assemble weighted normal equations:
    \[
    (\mathbf{A}^\top \mathbf{W} \mathbf{A}) \boldsymbol{\theta} = \mathbf{A}^\top \mathbf{W} \mathbf{f}
    \]
    
    \State Solve for coefficients $\boldsymbol{\theta}$
    
    \State \Return $T(\mathbf{x}) = \max(\theta_0, 0)$
    
    \end{algorithmic}
\end{algorithm}

\end{document}